\documentclass[11pt,a4paper]{amsart}

\pdfoutput=1

\usepackage{geometry}                
\usepackage{graphicx}
\usepackage{amssymb}
\usepackage{epstopdf}
\usepackage{fullpage}
\usepackage{color}
\usepackage{float}
\usepackage{caption}
\usepackage[table]{xcolor}
\usepackage{tablefootnote}

\usepackage{url}

\RequirePackage[OT1]{fontenc}
\RequirePackage{amsthm,amsmath}
\usepackage{amsmath, amssymb, amsfonts} 
\usepackage{algorithm}
\usepackage{algorithmic}
\usepackage{booktabs}  
\usepackage{graphicx, subfigure} 
\usepackage{multirow} 

\usepackage{mathabx}

\usepackage{longtable}

\newtheorem{proposition}{Proposition}[section]

\renewcommand{\Pr}{\mathsf{P}}
\newcommand{\p}{\mathsf{P}}

\newcommand{\M}{\mathsf{M}}

\newcommand{\bern}{\mathsf{Bernoulli}}

\renewcommand{\Pr}{\mathsf{P}}

\DeclareMathOperator*{\di}{\mathrm{d}\!}

\begin{document}

\title{Modeling association in microbial communities with clique loglinear models}

\author{Adrian Dobra\textsuperscript{1}}
\address{Department of Statistics, Department of Biobehavioral Nursing and Health Informatics, Center for Statistics and the Social Sciences and Center for Studies in Demography and Ecology, University of Washington, Box 354322, Seattle, WA 98195}
\email{adobra@uw.edu}

\author{Camilo Valdes\textsuperscript{1}}
\address{School of Computing and Information Sciences, Florida International University, 11200 SW 8th Street, ECS 354, Miami, FL 33199}
\email{hobbes182@gmail.com}

\author{Dragana Ajdic}
\address{Department of Dermatology and Cutaneous Surgery, and Department of Microbiology and Immunology, Miller School of Medicine, University of Miami, RMSB Room 2089, 1600 NW 10th Ave., Miami, FL 33136}
\email{d.ajdic@med.miami.edu}

\author{Bertrand Clarke}
\address{Department of Statistics, University of Nebraska-Lincoln, 340 Hardin Hall North Wing, Lincoln, NE 68583}
\email{bclarke3@unl.edu}

\author{Jennifer Clarke}
\address{Department of Statistics, and Department of Food Science and Technology, University of Nebraska-Lincoln, 340 Hardin Hall North Wing, Lincoln, NE 68583}
\email{jclarke3@unl.edu}

\footnote{A. Dobra and C. Valdes contributed equally to this work, and are joint first authors. Address correspondence to: \texttt{adobra@uw.edu}.}

\date{\today}                                           

\begin{abstract}
There is a growing awareness of the important roles that microbial communities play in complex biological processes. Modern investigation of these often uses next generation sequencing of metagenomic samples to determine community composition. We propose a statistical technique based on clique loglinear models and Bayes model averaging to identify microbial components in a metagenomic sample at various taxonomic levels that have significant associations. We describe the model class, a stochastic search technique for model selection, and the calculation of estimates of  posterior probabilities of interest. We demonstrate our approach using data from the Human Microbiome Project and from a study of the skin microbiome in chronic wound healing. Our technique also identifies significant dependencies among microbial components as evidence of possible microbial syntrophy. 
\end{abstract}

\maketitle

\tableofcontents

\section{Introduction}
\label{sec:intro}

Microbiomes  -- the communities of micro-organisms peculiar to specific environments such as mammalian skin or managed agricultural soil -- play key roles in a diverse set of biological phenomena, from plant growth to wine cultivation to human health and disease.  Metagenomics is the study of genetic material recovered
directly from a specific microbiome or environment without knowledge of the composition of the sample.
Thus, metagenomic-based studies generate valuable information about the composition of microbiomes
and differences in their composition that may be related to environmental differneces.  Traditionally,
 studying complex microbiome samples relied on intensive microbiological techniques involving the isolation and culturing of individual organisms followed by phenotypic or genotypic analysis. These techniques precluded microbial community profiling within a single sample. However, recent advances in high-throughput DNA sequencing technologies now permit whole-genome metagenomic sequencing (i.e., whole metagenome sequencing) without such isolation or culturing.  This means that characterization of complex microbial communities is now possible.

Whole metagenome sequencing has served as the primary tool for several high profile, collaborative research endeavors such as the U.S. National Institute of Health Human Microbiome Project \cite{nih-et-2009}, the U.S. Department of Energy Joint Genome Institute's Integrated Microbial Genomes (IMG) system \cite{markowitz-et-2014}, and the Canadian Institutes of Health Research Canadian Microbiome Initiative. Often, metagenome sequencing means that next generation sequencing (NGS) techniques are used.  These techniques differ from classical Sanger sequencing in that instead of sequencing a whole DNA molecule nucleotide by nucleotide, the sequencing is done in parallel at many points of the DNA molecule resulting in short reads, or simply reads, typically  ranging in length from 50 to 250 nucleotides.  Usually, a key step in the analysis of NGS data is aligning the reads to a collection of consensus sequences or reference genomes for a collection of organisms.  Whole metagenome sequencing is the general case for which our formal reasoning
is designed: our examples use whole genome and sequencing and 15S sequencing. The differences are addressed in Appendix, Section \ref{sec:sequencing}.

It is well known that compositional studies of microbiomes alone provide no information about potential symbiosis, or syntrophy --  settings in which the metabolic waste products from one microbe provide nutrients for another -- among species or strains \cite{levy-borenstein-2013}. Indeed, microbial communities in diverse settings have been shown to form syntrophic relationships. Such relationships have been posited to drive pathogenicity \cite{kievit-idlewski-2000,koch-et-2014}. A simple approach to infer possible syntrophic relationships is to examine rates of co-occurence of micro-organisms in the same habitat across samples \cite{hoffmann-et-2013}. However, these methods cannot be used with a single sample.  They rely on co-occurence across many samples. In addition, in most metagenomic studies based on sequencing there is a portion of sequencing reads that cannot be associated with any known microorganisms in a particular environment, and these reads are often discarded inappropriately.

To address these limitations, we introduce a statistical approach based on a class of loglinear models which we call clique loglinear models that can assess both association among bacteria within a single sample (or across samples), and the likelihood of a specific bacterium, including a previously unknown bacterium, being in the sample. Our 
focus is on whole metagenome sequencing, as our goal is to identify bacteria and associations 
among them at various taxonomic levels (e.g., genus, species, or strain).

There are two ways our methodology is novel.  First, the way the data are pre-processed for analysis as a multi-way contingency table is new.  In whole metagenome sequencing a collection of reads is sampled from a biological community within one sample.  We align these reads to a database of microbial reference genomes, and the result is a categorical dataset showing the reference genomes to which each read aligns. In these data, one row corresponds to one read, and one column corresponds to a genome -- indicating the genomes to which each read maps.  The rows are independent if the reads are from different organisms and, often nearly independent even when they are from the same organism.   

Although initially counter-intuitive, this is seen empirically in a Bayesian context in \cite{Clarke:etal:2015}. In fact, assuming independence among a large number of reads is a reasonable first approximation because i) the number of nucleotides in the DNA molecules is very large so dependence will be rare, and ii) even when 
reads are regarded as dependent this rests partially on their gene products. Here we are not looking at gene products so dependence among them is irrelevant, making the dependence among reads smaller than one would initially expect. For the present we ignore the dependence among reads as a pragmatic approximation, and return to it in Section \ref{sec:discussion}. Hence, our procedure assumes that there are $B$ genomes and $R$ reads then we have $R$ outcomes of a categorical random variable that assumes one of $2^B$ values representing the possible patterns of align or not align that each read has for each genome. The details are given in Section \ref{sec:methodology}.

Second, we introduce clique loglinear models to search for associations among bacterial strains or other taxa using the reads. We do this by defining a stochastic search procedure, and a Bayes model average (BMA) that combines the most relecvant clique loglinear models found by the search. To represent the associations, we produce connectivity graphs showing which bacterial genera (or other taxonomic unit) are related by higher order terms. This is possible because a clique loglinear model is a compound of disjoint collections of higher order terms, each collection permitting all possible interactions amongst the categories at the taxonomic level under study.  Overall, clique loglinear models are a sparse subset of all hierarchical loglinear models \cite{bishop-et-2007}, and this is operationally satisfactory since the associations among bacterial strains are often sparse as well.  Otherwise put, the class of clique loglinear models is small enough to be tractable, yet large enough to be used for data summarization and model selection. Given the increasing speed of computing and accumulating knowledge
about which bacteria are in which microbiome, this task is likely to be easier in the future than it is now.

Fundamentally, in what follows, clique loglinear models are not proposed as  physical models for the interactions between genomes or other taxonomic levels, except possibly for a few narrow settings. Indeed, valid models for joint distributions would likely be dynamic as well as more complicated than clique loglinear models permit.  Instead, here, clique loglinear models are the basis of a search strategy for relationships among genomes as encapsulated by the higher order terms in the models.  Our evidence supports the supposition that the cliques found by clique loglinear models are present, possibly the ones most strongly present, even if the collection of associations they represent is incomplete.  Hence, we are de facto using models as if they were summary statistics for a data set rather than as a statement about the real phenomenon, which is often too complicated to model at present. We argue that, as summary statistics, clique loglinear models capture enough information in the data that the results of the search strategy are useful.

The structure of this paper is as follows.  In Section \ref{sec:methodology} we describe how our data are prepared for analysis, and formally define the clique loglinear model class.  We also discuss several existing ways to analyze the NGS data that, unfortunately, do not extend to large values of $B$. This leads to introducing a stochastic search procedure that enables us to compute the various quantities of interest so we can make inferences about the presence of various strains, species, and genera of bacteria and their associations
within a given taxonomic level. Section \ref{sec:sims} presents a series of simulations to verify that our methodology qualitatively generates the results one would anticipate. In Sections \ref{sec:analyses} and \ref{sec:footwound} we analyze two datasets, and interpret our results in their scientific contexts.  For the first of these our results are consistent with the findings from a more traditional  approach to analysis of the same data. For the second, we generate results that seem plausible given the experimental context; there is  no previous analysis for comparison purposes. This shows that our method provides an alternative to expensive laboratory work. Finally, in Section \ref{sec:discussion} we discuss how several features of our formalism relate to the real biological questions we have addressed.

\section{Analyzing NGS data using clique loglinear models}
\label{sec:methodology}

In this section we motivate and outline our overall methodology for using metagenomic NGS data
to detect associations among, say, bacterial strains or genera.  In order, we explain i) the 
pre-processing of the NGS data into contingency tables (Section \ref{sec:repdata}), ii) what clique loglinear models are (Section \ref{sec:cliqueloglin}), iii) existent methods for loglinear model determination (Section \ref{sec:largeB}), and iv) our analytical methodology (Section \ref{sec:ourmeth}). 

\subsection{Representing NGS data as a sparse contingency table}
\label{sec:repdata}

Suppose there are $r=1, \ldots ,R$ reads and $b=1, \ldots , B$ known bacterial genomes. Because of the 
way sequencing reads are generated, they can match none, one, or several bacterial strains on a list
$\{C_1, \ldots , C_B\}$.  Table \ref{tab:actualData} shows a hypothetical example of such a data set. 
The first read only matches bacterial strain 1, while reads $2, \ldots, R$ each match at least two strains.
Once the data have been put in the form of Table \ref{tab:actualData},
the patterns of matches and non-matches define candidate interactions involving two, 
three or more strains and can be regarded as generated by $B$ binary categorical variables 
each evaluated at one of the $R$ reads.

Now, the sample of $R$ reads from a metagenomic population of bacterial genomes and the list 
of bacterial genomes can be represented by a
$R \times B$ matrix $(c_{rb})_{RB}$ that we call a connectivity matrix, in which
\begin{eqnarray}
c_{rb} =
\begin{cases}
 1, & \hbox{read} ~r~ \hbox{aligns to strain} ~ b, \\
 0, & \hbox{read} ~ r~ \hbox{does not align to strain} ~ b .
 \end{cases} 
 \nonumber
 \end{eqnarray}
Let $\mathcal{B}=\{1,2,\ldots,B\}$. Each row may be regarded as a vector valued outcome of the vector valued random variable  $\mathbf{X}_{\mathcal{B}} = (X_1, \ldots, X_B)$ in which $X_b=X_b(r)$ is the indicator variable for 
a sampled read $r$ to align (or match) to genome $b$.  Each outcome of $\mathbf{X}_{\mathcal{B}}$ assumes one of $2^B$ patterns of zeroes and ones in $\mathcal{X}_{\mathcal{B}} = \{ 0,1\}^B$.  These vectors of length $B$
generate a $B$-dimensional contingency table $\mathbf{n}_{\mathcal{B}}$ in which the count $\mathbf{n}_{\mathcal{B}}(\mathbf{x}_{\mathcal{B}})$ in cell $\mathbf{x}_{\mathcal{B}}\in \mathcal{X}_{\mathcal{B}}$ gives the number of reads that share the same pattern of alignments to the $B$ genomes. The case $B=3$ is shown in 
Figure \ref{fig:contingency_table}.    

We want to model the joint distribution of $\mathbf{X}_B$ to obtain estimates 
of interesting cell probabilities $\p_{\mathcal{B}}(\mathbf{x}_{\mathcal{B}}) = \Pr(\mathbf{X}_{\mathcal{B}} = \mathbf{x}_{\mathcal{B}})$, and relations amongst them. For instance, 
\begin{align} 
\label{eq:probunknown}
\Pr(X_{1}=0,\ldots,X_{B}=0)
\end{align}
\noindent is the probability that a sampled read aligns to none of the $B$ reference genomes. If 
the estimate of \eqref{eq:probunknown} is high, we might infer that we
have found a bacterium or other microbial source not amongst the $C_b$'s. 
By contrast,
\begin{align} 
\label{eq:probshortsequence}
\Pr(X_{b^{*}}=1,\{ X_{b}=0,\ \forall \ b\ne b^{*}\})
\end{align}
\noindent is the probability that a sampled read comes from $C_{b*}$ and does {\it not} come from any of the other $(B-1)$ genomes. To identify the bacteria that are most likely to be present, it is natural to pick the $C_b$'s with the highest values of \eqref{eq:probshortsequence}. 

\begin{table}[htp]
\begin{center}
\begin{tabular}{|c|cccc|}\hline
 Read & Genome 1 & Genome 2 & $\cdots$ & Genome $B$ \\ \hline
  1 & {\it Match} & No match & $\cdots$ & No match\\
  2 & {\it Match} & {\it Match} & $\cdots$ & No match\\
  3 & {\it Match} & No match & $\cdots$ & {\it Match}\\
$\vdots$ & $\vdots$ & $\vdots$ & $\cdots$ & $\vdots$\\
  $R$ & {\it Match} & {\it Match} & $\cdots$ & {\it Match}\\ \hline
\end{tabular}
\caption{\label{tab:actualData} In this $R \times B$ data matrix,
each of the $r=1, \ldots , R$ rows represents a short read and is 
regarded as a data point.  The $b=1, \ldots , B$ entries for each row represent whether or not the short read $r$ matches the genome $b$
represented by the column.}
\end{center}
\end{table}%

\begin{figure}
\centering
\centerline{
\includegraphics[width = 2in]{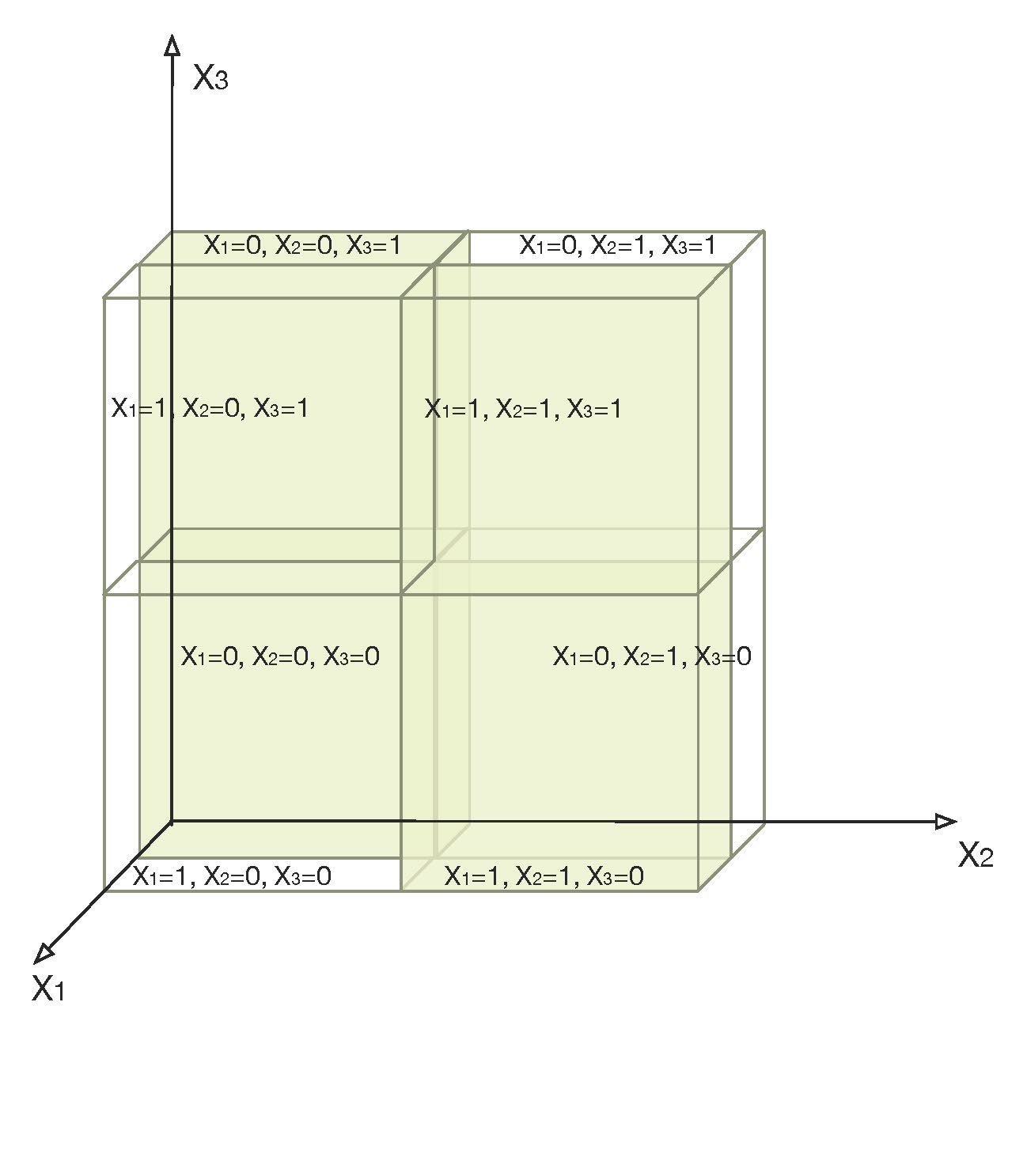}
}
\caption{\label{fig:contingency_table}An example contingency table formed by $B$=3 bacterial genomes. The table is $2\times 2\times 2$ and has $2^3 = 8$ cells, i.e., distinct vectors of $0$s and $1$s of length 3. Each cell contains the count of the number of reads whose vector $(X_1,X_2,X_3)$ matches the pattern shown.}
\end{figure}
	
Once the data form a connectivity matrix as in Figure \ref{fig:contingency_table}, loglinear models \cite{bishop-et-2007} are well suited to represent associations among bacterial taxa in a community. Loglinear models are flexible, interpretable, and describe the dependence of cell counts on the component categorical variables (or bacterial taxa, in our case). Selection and estimation for loglinear models has been well studied in the literature \cite{edwards-havranek-1985,whittaker-1990}, but satisfactory answers for high-dimensional contingency tables remain elusive. Nevertheless, large values of $B$ arise naturally in metagenomics and the subset of loglinear models that we define here may be a promising way to begin answering questions about the interactions within taxonomic levels.

\subsection{Clique loglinear models}
\label{sec:cliqueloglin}

For a set $C\subseteq \mathcal{B}$, we denote $\mathcal{X}_C = \{0,1\}^{|C|}$, where $|C|$ stands for the number of elements of $C$. The subvector $\mathbf{X}_C$ of $\mathbf{X}_{\mathcal{B}}$ takes values $\mathbf{x}_C\in \mathcal{X}_C$. The $C$-marginal $\mathbf{n}_C$ of $\mathbf{n}_{\mathcal{B}}$ has cell counts $\mathbf{n}_C(\mathbf{x}_C)=\sum_{\mathbf{x}_{\mathcal{B}\setminus C}} \mathbf{n}_{\mathcal{B}}(\mathbf{x}_C,\mathbf{x}_{\mathcal{B}\setminus C})$. The corresponding marginal cell probabilities are $\p_C(\mathbf{x}_C) = \Pr(\mathbf{X}_C = \mathbf{x}_C)$.

Consider a hierarchical loglinear model $\M$ with $k$ generators $\mathcal{C}({\M})=\{ C_1,C_2,\ldots,C_k\}$, where $C_j\subseteq \mathcal{B}$, for $j=1,\ldots,k$, and $k\ge 1$ \cite{bishop-et-2007,edwards-havranek-1985}. Under this model, the cell probabilities associated with $\mathbf{X}_{\mathcal{B}}$ are represented as \cite{whittaker-1990}:
\begin{eqnarray}\label{eq:loglindec}
 && \log \p_{\mathcal{B}}(\mathbf{x}_{\mathcal{B}}) = u_{\emptyset}+\sum\limits_{\{C: \emptyset \ne C\subseteq C_j ~\hbox{for some}~j\in \{1,\ldots,k\}\}} u_C(\mathbf{x}_C).
\end{eqnarray}
Here $u_{\emptyset}$ is an intercept, and $\{u_C(\mathbf{x}_C):\mathbf{x}_C\in \mathcal{X}_C\}$ is the $|C|$-way interaction associated with the subvector $\mathbf{X}_C$ of $\mathbf{X}_B$. This model can be made identifiable either by imposing the sum to zero constraints $\sum_{\mathbf{x}_C\in \mathcal{X}_C} u_C(\mathbf{x}_C) = 0$, or by imposing the baseline equal with zero constraints that set $u_C(\mathbf{x}_C) = 0$ if one element of $\mathbf{x}_C$ is zero. For the latter, the loglinear expansion \eqref{eq:loglindec} becomes:
\begin{eqnarray}\label{eq:loglindecidentifiable}
 && \log \p_{\mathcal{B}}(\mathbf{x}_{\mathcal{B}}) = u_{\emptyset}+\sum\limits_{\{C: \emptyset \ne C\subseteq C_j ~\hbox{for some}~j\in \{1,\ldots,k\}\}} u_C\prod\limits_{i\in C}x_i,
\end{eqnarray}
\noindent where $u_C = u_C(1,\ldots,1)$.

A hierarchical loglinear model $\M$ is a clique loglinear model if its generators form a partition of $\mathcal{B}$: $\bigcup_{j=1}^k C_j = \mathcal{B}$, $C_{j_1}\cap C_{j_2} = \emptyset$ for $j_1\ne j_2$. In this case, the cell probabilities \eqref{eq:loglindecidentifiable} are written as:
\begin{eqnarray}\label{eq:cliqueloglin}
 && \log \p_{\mathcal{B}}(\mathbf{x}_{\mathcal{B}}) = u_{\emptyset}+\sum\limits_{j=1}^k\sum\limits_{\{C: \emptyset \ne C\subseteq C_j\}} u_C\prod\limits_{i\in C}x_i.
\end{eqnarray}
Thus, under a clique loglinear model, the log cell probabilities are decomposed as a sum of groups of interaction terms in which each group represents a collection of categorical variables that may interact with each other in all possible ways, but do not interact at all with categorical variables in other groups. Formally, the interpretation of clique loglinear models comes from this result:
\begin{proposition} \label{prop:cliqueloglin}
 Let $D_1$ and $D_2$ be two subsets of $\mathcal{B}$ that are also subsets of two different generators of a clique loglinear model $\M$. Then the random subvectors $\mathbf{X}_{D_1}$ and $\mathbf{X}_{D_2}$ are independent.
\end{proposition}
\begin{proof}
 We collapse across the levels of $\mathbf{X}_{\mathcal{B}\setminus (D_1\cup D_2)}$ in the loglinear expasion \eqref{eq:cliqueloglin}. The marginal cell probabilities associated with $\mathbf{X}_{D_1\cup D_2}$ have the form:
$$
 \log \p_{D_1\cup D_2}(\mathbf{x}_{D_1},\mathbf{x}_{D_2}) =  u_{\emptyset}+ \sum\limits_{\{C: \emptyset \ne C\subseteq D_1\}} u_C\prod\limits_{i\in C}x_i + \sum\limits_{\{C: \emptyset \ne C\subseteq D_2\}} u_C\prod\limits_{i\in C}x_i.
$$
Since the first term is a constant, the second term is a function of the levels of $\mathbf{X}_{D_1}$, and the third term is a function of the levels of $\mathbf{X}_{D_2}$, it follows that $\mathbf{X}_{D_1}$ and $\mathbf{X}_{D_2}$ are indeed independent.
\end{proof}
A consequence of Proposition \eqref{prop:cliqueloglin} is that the cell probabilities of $\M$ decompose as a product of marginal cell probabilities associated with its generators:
\begin{eqnarray} \label{eq:probdecomp}
  \p_{\mathcal{B}}(\mathbf{x}_{\mathcal{B}}) = \prod\limits_{j=1}^k \p_{C_j}(\mathbf{x}_{C_j}).
\end{eqnarray}
We denote by $\mathbf{u}_{\M}$ all the interaction terms that appear in \eqref{eq:cliqueloglin}. Under Multinomial sampling, the log-likelihood function is written as a function of the interaction terms as follows:
\begin{eqnarray} \label{eq:loglik}
 l(\mathbf{u}_{\M},\mathbf{n}_{\mathcal{B}}) = R u_{\emptyset} + \sum\limits_{j=1}^k\sum\limits_{\{C: \emptyset \ne C\subseteq C_j\}} u_C \mathbf{n}_C(\mathbf{x}_C) \prod\limits_{i\in C}x_i.
\end{eqnarray}
By using Lagrange multipliers in \eqref{eq:loglik} and \eqref{eq:probdecomp} \cite{whittaker-1990}, it can be shown that the MLEs of the cell probabilities under $\M$ are
\begin{eqnarray} \label{eq:mlescliqueloglin}
 \widehat{p} _{\mathcal{B}}(\mathbf{x}_{\mathcal{B}}) = R^{-k} \prod\limits_{j=1}^k \mathbf{n}_{C_j}(\mathbf{x}_{C_j}).
\end{eqnarray}
Equation \eqref{eq:mlescliqueloglin} shows that the MLEs of the cell probabilities of a clique loglinear model exist if and only if the counts in the marginal tables associated with its generators are strictly positive. This existence criterion is easily applicable in a computational efficient manner. By contrast, determining the existence of the MLEs for arbitrary hierarchical loglinear models is a difficult problem that has been solved theoretically \cite{fienberg-rinaldo-2007}. However, at the present time, there do not seem to exist any implementable algorithms for assessing the existence of MLEs of hierarchical loglinear models that are also computationally efficient when the number $B$ of categorical variables involved is large.

We named this class of loglinear models based on the representation of the interaction structure defined by the u-terms $\mathbf{u}_{\M}$ as an independence graph \cite{whittaker-1990}. This is an undirected graph $G$ with vertices $\mathcal{B}$ and set of edges $E$. Each element $b\in \mathcal{B}$ is associated with the component $X_b$ of $\mathbf{X}_{\mathcal{B}}$.  An edge $e=(b_1,b_2)$ is included in $E$ if there is a generator $C_j$ of $\M$ such that $\{ b_1,b_2\}\subseteq C_j$. Proposition \eqref{prop:cliqueloglin} implies that the independence graph $G$ of a clique loglinear model $\M$ has a special structure: the generators of $\M$ are the connected components of $G$, and are also maximal complete subgraphs or cliques \cite{lauritzen-1996}. As such, the independence graph of a clique loglinear model is obtained by putting together complete subgraphs without adding any edge between them. These cliques are the generators of the loglinear model, and uniquely identify it.

The class of clique loglinear models is a subset of decomposable loglinear models, which, in turn, is a subset of graphical loglinear models that are themselves a subclass of hierarchical models \--- see Appendix, Section \ref{sec:nested}. The restriction to clique loglinear models offers key computational advantages: in addition to an easy way to calculate the MLEs and check their existence, these models are straightforward to interpret (Proposition \eqref{prop:cliqueloglin}), and allow the development of computationally efficient model determination algorithms that scale well when $R$ or $B$ become large.

The number of clique loglinear models for $B$ categorical variables is the number of decomposition of $B$ into integers \cite{abramowitz-stegun-1972}:
\begin{eqnarray} \label{eq:numbercliqueloglin}
 \mathcal{P}(B) = \frac{1}{\pi\sqrt{2}}\sum_{j=1}^{\infty}\sqrt{j}A_j(B)\frac{\di}{\di n} \frac{\sinh\left( \frac{\pi}{j}\sqrt{\frac{2}{3}\left( B-\frac{1}{24}\right)}\right)}{\sqrt{B-\frac{1}{24}}},
\end{eqnarray}
where 
$$
 A_j(B) = \sum\limits_{0< h\le j, (h,j)=1} e^{\pi i \left( s(h,j) - \frac{2 h B}{j}\right)}, \quad s(h,j) = \sum_{l=1}^{j-1} \frac{l}{j} \left(\left( \frac{hl}{j}\right)\right),
$$
with $((x)) = x-[x]-\frac{1}{2}$ if $x$ is an integer, and $0$ otherwise, and $(h,j)$ is the greatest common divisor of $h$ and $j$. For example, $\mathcal{P}(100) = 190,569,292$, $\mathcal{P}(200) \approx 3.973e+12$ and $\mathcal{P}(1000) \approx 1.321e+19$ \cite{hankin-2006}. Therefore,  although this is the smallest class of hierarchical loglinear models, it still contains a significantly large number of possible models that allow modeling various patterns of interactions among many categorical variables.

\subsection{Loglinear model selection methods}
\label{sec:largeB}

Capturing associations can be done by determining models for the joint distributions of the observed categorical variables while recognizing that these random variables do not vary independently of each other. However, the association structure within a microbial community is likely to be sparse because most of the possible model higher order terms are likely to be discarded. This happens because, given the length of the list of reference genomes, most bacteria only occur jointly with a relatively small number of other bacteria. Given this, we argue that classes of hierarchical loglinear models are suitable for representing multivariate associations in high-dimensional sparse contingency tables.  However, existent methods for loglinear model selection do not seem to be effective for such tables.  

Severe difficulties can arise from the sparsity of high-dimensional tables \cite{fienberg-rinaldo-2007} such as the invalidation of asymptotic approximations to the  null distribution of the generalized likelihood ratio test statistic.  For this reason we use a Bayes formulation that avoids these issues through the specification of prior distributions for model parameters \cite{clyde-george-2004}. The limitation is that prior selection for model selection problems is known to be difficult and sometimes controversial. We try to evade this problem by defaulting to a flat prior in the hope that the data will be sufficiently informative as to generate an informative posterior.

The key difficulty is the size of the space of possible loglinear models.  For example, when $B=5$ the number of possible hierarchical loglinear models is $7580$; for $B=8$ variables this number increases to $5.6\times 10^{22}$ \cite{dellaportas-forster-1999}. In a Bayesian framework, one approach for the analysis of higher dimensional contingency tables is called copula Gaussian graphical models \cite{dobra-lenkoski-2011}, and it has successfully been used to analyze a 16-dimensional table. More recently, ultra-sparse high-dimensional contingency tables 
have been analyzed using probabilistic tensor factorizations induced through a Dirichlet process mixture model of product multinomial distributions \cite{dunson-xing-2009,canale-dunson-2011,bhattacharya-dunson-2012,kunihama-dunson-2013}. These papers present simulation studies and real-world data examples that involve up to 50 categorical variables. 
While promising, Bayesian sparse tensor factorization methods have yet to be applied to categorical data sets with more variables, say $70$, which is the setting of greatest interest for the problems we address here.

Because the space of possible models is extremely large, various stochastic search schemes have been used to identify models with high posterior probability. \cite{dellaportas-forster-1999} is a key reference, although there are other papers that develop stochastic search schemes for discrete data \cite{madigan-raftery-1994,madigan-york-1995,madigan-york-1997,tarantola-2004,dellaportas-tarantola-2005,dobra-massam-2010}. One feature of these and other stochastic searches on spaces of hierarchical, graphical and 
decomposable loglinear models  \--- see, for example, \cite{massam-et-2009} \--- is that these models involve repeated transitions from one model to another model. This necessitates ensuring the next model is still in the target space of models. For instance, for decomposable graphs, transitioning from a current decomposable graph to another decomposable graph involves checks that the decomposability property is preserved. While such checks can be done relatively quickly for graphs with few vertices, for graphs that involve hundreds of
vertices the running time of stochastic searches increases rapidly.

Since considerable computational effort is required to visit loglinear models sequentially by adding and removing higher order terms, restricting the model space to, say, clique loglinear models provides a necessary reduction in the running time of model determination algorithms. This makes the required computations intensive yet feasible. 

\subsection{Our stochastic search method}
\label{sec:ourmeth}

Let $\M$ vary over the collection $\mathcal{M}$ of $B$-dimensional clique loglinear models for which the MLEs exist. We want to find $\M$'s that fit the data well and are parsimonious. For this purpose we choose the Bayesian information criterion (BIC). For large sample sizes, it is well known that the BIC is an approximation to the mode of a posterior distribution over a model space. The BIC is also optimal in a Bayes testing sense \cite{schwarz-1978}. The calculation of BIC for clique loglinear models proceeds as follows. Denote by $\mathcal{C}({\M})=\{ C_1,C_2,\ldots,C_k\}$ the generators of $\M$. The MLEs of the mean cell values under $\M$ are calculated based on \eqref{eq:mlescliqueloglin}:
\begin{eqnarray}\label{eq:mlescells}
 &&  \log \widehat{m}_{\mathcal{B}}(\mathbf{x}_{\mathcal{B}}) = \log (R \widehat{p} _{\mathcal{B}}(\mathbf{x}_{\mathcal{B}})) = \sum\limits_{j=1}^k \log \mathbf{n}_{C_j}(\mathbf{x}_{C_j}) - (k-1)\log R,
\end{eqnarray}
\noindent for all $\mathbf{x}_{\mathcal{B}} \in \mathcal{X}_{\mathcal{B}}$. From \eqref{eq:cliqueloglin} we see that the number of free interaction terms that appear in $\M$ is equal with the sum of the number of nonempty subsets of the generators of $\M$. Therefore the BIC of $\M$ is given by
\begin{eqnarray} \label{eq:bicformula}
 BIC(\M) & = & -2\sum\limits_{\{\mathbf{x}_{\mathcal{B}} \in \mathcal{X}_{\mathcal{B}} :\mathbf{n}_{\mathcal{B}}(\mathbf{x}_{\mathcal{B}})>0\}} \mathbf{n}_{\mathcal{B}}(\mathbf{x}_{\mathcal{B}}) \log \widehat{m}_{\mathcal{B}}(\mathbf{x}_{\mathcal{B}})\\ \nonumber &&  + \left( \sum\limits_{j=1}^k 2^{|C_j|} - k + 1\right)  \log R.
\end{eqnarray}
Equations \eqref{eq:mlescells} and \eqref{eq:bicformula} show that the BIC of a clique loglinear model can be efficiently calculated even for large contingency tables since no iterative numerical optimization methods are involved as it would have been the case for arbitrary graphical and hierarchical loglinear models. The calculation of the log mean cell values can also be performed using a formula for decomposable loglinear models \cite{lauritzen-1996}, but the calculation of the number of free interaction terms of these models would have been complicated by their overlapping sets of generators. For this reason, the calculation of BIC for clique loglinear models is easier as compared to any other loglinear model that does not belong to this class.

Consider the following distribution over $\mathcal{M}$:
\begin{align} 
\label{eq:targetdist}
 \pi(\M) \propto \exp( -BIC(\M)).
 \end{align}
Finding clique loglinear models with smaller values of BIC is equivalent to finding models at or close to the modes of the distribution \eqref{eq:targetdist}. We can think of $\pi(\M)$ as a posterior distribution over $\mathcal{M}$ obtained by assuming a flat prior over $\mathcal{M}$. Thus, the $\pi(\M)$'s can be considered to be the Bayes model weights, and, in the sequal, these weights will be used to perform model averaging using Occam's window. This methodology originates in \cite{madigan-raftery-1994}, and has been developed in numerous other contexts, e.g., dynamic linear models \cite{onorante-raftery-2016} and graphical models \cite{madigan-et-1995}. For a more in depth discussion, see Supplementary Materials, Section \ref{sec:reallybayes}.

Our goal is to find clique loglinear models that have large posterior weights \eqref{eq:targetdist}. The largest would  achieve
\begin{align}
\widehat{\M}_{\sf opt} = \arg ~ \max_{{\mathcal{M}}} \pi(\M).
\label{eq:actual}
\end{align}
However, models that have posterior weights comparable to that of the optimal model $\widehat{\M}_{\sf opt}$ are also relevant. The stochastic search algorithm we propose below is devised to seek the set of models
\begin{eqnarray}\label{eq:targetset}
 \mathcal{S}(c) = \left\{ \M\in \mathcal{M}: \pi(\M)\ge c \pi(\widehat{\M}_{\sf opt})\right\},
\end{eqnarray}
\noindent where $c\in (0,1)$ is a constant that needs to be specified before the start of the algorithm. The clique loglinear models that do not belong to $\mathcal{S}(c)$ are discarded. The idea of eliminating models with low posterior probability compared to the highest posterior probability model is based on the Occam's window principle of \cite{madigan-raftery-1994}.
 
For ease of exposition we begin by stating our procedure informally.   Our stochastic search procedure moves towards models with larger values of $\pi(\M)$. The models that are visited in a run are collected as if in a bag. Each run of the stochastic search algorithm collects models until it appears to reach a local optimum. At that point the stochastic search algorithm will likely visit only models that are already in the bag. We use many different runs, and combine all the bags of models collected in each run into a larger bag $\mathcal{S}$. Out of this bag, we only retain those models that have comparable posterior weights with the best model identified across all runs:
\begin{eqnarray}\label{eq:targetsetapprox}
 \widehat{\mathcal{S}}(c) = \left\{ \M\in \mathcal{S}: \pi(\M)\ge c \max\limits_{\M\in \mathcal{S}} \pi(\M)\right\}.
\end{eqnarray}
Across multiple runs that were sufficiently long, we would hope that $\widehat{\mathcal{S}}(c)$ from \eqref{eq:targetsetapprox} will approximate well $\mathcal{S}(c)$ in \eqref{eq:targetset}. This is very likely to happen if $\widehat{\M}_{\sf opt}$ has been visited and included in $\mathcal{S}$. An empirical test for figuring out whether $\widehat{\M}_{\sf opt}$ was indeed identified is to determine the proportion of runs that reached $\arg ~ \max_{{\mathcal{S}}} \pi(\M)$. A high proportion of runs that ended up visiting the best model in $\mathcal{S}$ represents a good indication that $\widehat{\M}_{\sf opt}$ might indeed be in $\mathcal{S}$. In the sequel, we perform Bayes model averaging using the models in $\widehat{\mathcal{S}}(c)$ with weights in \eqref{eq:targetdist}, and this lets us estimate the quantities of interest. Models not in $\widehat{\mathcal{S}}(c)$ are discarded; this is justified if $\widehat{\mathcal{S}}(c)$ comprises most models that have large posterior probabilities.   

Our stochastic algorithm for identifying $\mathcal{S}(c)$ from \eqref{eq:targetset} proceeds as follows.  We start with a randomly generated clique model. If any of the marginals associated with the generators of this model contain counts of zero, the MLEs of this model do not exist and another random model is generated. We repeat these steps until a valid clique loglinear model is generated; we denote this model with $\M_{0}$. Starting with $\M_{0}$ we generate a chain of models $\langle \M_{t} \rangle$ for $t =1, 2, \ldots$. At step $t$, with equal probability, we select one of the following four ways of producing a valid (i.e., for which the MLEs \eqref{eq:mlescliqueloglin} exist) candidate clique loglinear model $\M^{\prime}$:

\noindent {\it (i)} Split a random clique of $\M_t$ into two cliques.\\ 
\noindent {\it (ii)} Join two random cliques of $\M_t$ into a clique. \\
\noindent {\it (iii)} Switch two random elements that belong to two random cliques of $\M_t$.\\
\noindent {\it (iv)} Move a random element of a random clique of $\M_t$ to another random clique of $\M_t$. 

After sampling a move of type (i), (ii), (iii) or (iv), we produce a clique loglinear model $\M^{\prime}$ by applying a move of that type to model $\M_{t}$. For moves of type (i), we uniformly select a clique of $\M_{t}$ and divide the elements in that clique $\mathcal{C}^{\prime}$ into two disjoint sets $\mathcal{C}^{\prime}_1$ and $\mathcal{C}^{\prime}_2$ that become two new cliques of  $\M^{\prime}$. The other cliques of  $\M_{t}$ are also cliques for $\M^{\prime}$. For moves of type (ii), we uniformly sample two cliques $\mathcal{C}^{\prime}_1$ and $\mathcal{C}^{\prime}_2$ of $\M_{t}$, and form a new clique $\mathcal{C}^{\prime}=\mathcal{C}^{\prime}_1\cup \mathcal{C}^{\prime}_2$. The other cliques of $\M_{t}$ together with $\mathcal{C}^{\prime}$ are the cliques of $\M^{\prime}$. For moves of type (iii), we uniformly sample two cliques $\mathcal{C}^{\prime}_1$ and $\mathcal{C}^{\prime}_2$ of $\M_{t}$, and also uniformly sample a element $v_1\in \mathcal{C}^{\prime}_1$ and a element $v_2\in \mathcal{C}^{\prime}_2$. We form two new cliques $\mathcal{C}^{\prime\prime}_1=\mathcal{C}^{\prime}_1\setminus \{v_1\} \cup \{v_2\}$ and $\mathcal{C}^{\prime\prime}_2=\mathcal{C}^{\prime}_2\setminus \{v_2\} \cup \{v_1\}$. The cliques $\mathcal{C}^{\prime\prime}_1$, $\mathcal{C}^{\prime\prime}_2$ together with the other cliques of $\M_{t}$ give the candidate model $\M^{\prime}$. For moves of type (iv), we uniformly sample two cliques $\mathcal{C}^{\prime}_1$ and $\mathcal{C}^{\prime}_2$ of $\M_{t}$, and also uniformly sample a element $v_1\in \mathcal{C}^{\prime}_1$. We form two new cliques $\mathcal{C}^{\prime\prime}_1=\mathcal{C}^{\prime}_1\setminus \{v_1\}$ and $\mathcal{C}^{\prime\prime}_2=\mathcal{C}^{\prime}_2\cup \{v_1\}$. The cliques $\mathcal{C}^{\prime\prime}_1$, $\mathcal{C}^{\prime\prime}_2$ together with the other cliques of $\M_{t}$ give the candidate model $\M^{\prime}$. 

If the MLEs of $\M^{\prime}$ do not exist, we set $\M_{t+1}=\M_{t}$. If the MLEs of  $\M^{\prime}$ exist, we set $\M_{t+1}=\M^{\prime}$ with probability $\min \{1,\pi(\M^{\prime})/\pi(\M_{t})\}$. Otherwise we set $\M_{t+1}=\M_{t}$. This stochastic search algorithm typically moves to models with larger $\pi(\M)$'s. If the sampled candidate model  $\M^{\prime}$ happens to have a smaller $\pi(\M)$ than the current model $\M_{t}$, the algorithm could still visit it with positive probability. This is useful because sometimes models with smaller $\pi(\M)$s must be visited before finding models with larger $\pi(\M)$'s.  This is the case when $\M$ is a local maximum but not a global maximum. The geometry of the space of models affects this:  getting stuck in a local maximum is not a problem if it is a global maximum; on the other hand,  the model space is discrete so it is possible that the models with the largest $\pi(\M)$'s are not very similar to each other.

Strictly speaking, only moves of type (i) and (ii) are needed to connect a clique loglinear model for which the MLE exists with any other clique loglinear model in $\mathcal{M}$. However, in computational results not included here, we found that an algorithm that included moves of types (iii) and (iv) was less likely to get stuck in local maxima of $\pi(\cdot)$. We note that this is a stochastic search procedure, not a Markov Chain Monte Carlo procedure so that the acceptance probabilities are not relevant; see the Supplementary Materials, Section \ref{sec:whatabout} for a discussion of this point. Furthermore, sparse contingency tables such as we are studying here frequently have unbalanced counts. This is rarely a problem for the BIC, as discussed in the Supplementary Materials, Section \ref{sec:skewed}.

\section{Simulation results}
\label{sec:sims}

To benchmark the performance of the method, we created a synthetic experiment with a known community dependency structure.  We obtained 2,273 bacterial genomes from the  National Center for Biotechnology Information (NCBI) GenBank database. These genomes were collected from GenBank's complete genome set, or those genomes that are considered to have a final DNA sequence per genomic structure (chromosomes and/or plasmids). From these 2,273 genomes, we randomly chose 200 genomes and created a population connectivity matrix representing 1,000 synthetic genomic reads that indicates the connectivity among the genomes.  As seen in Figure \ref{fig:sim_fig1}, each simulated read has a corresponding row in the connectivity matrix with a match for at least one genome; this is indicated by ``1''.   The matrix is based on a  file supplied to the simulation program that indicates which genomes are present and what cliques they form. If two genomes are in the same clique they are given 1s for 80\% of their joint reads (as assigned by i.i.d. $\bern(0.8)$ random variables). The remaining cells in the $1000 \times 200$ connectivity matrix are randomly filled with 0s and 1s sampled from a $\bern(0.2)$ distribution. This procedure gives a connectivity matrix consistent with a chosen clique structure on genomes.  Note that not all 200 genomes are shown because only some were in nontrivial cliques. Further details are given in the Supplementary Materials, Section \ref{sec:sims-supp}.

\begin{figure}
\centering
\centerline{
\includegraphics[width = 4in]{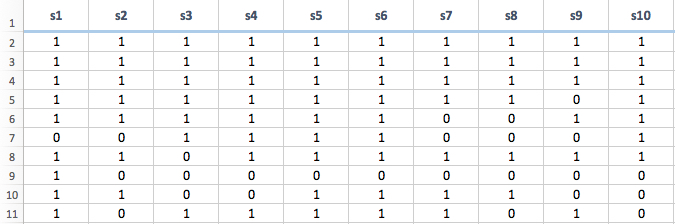}
}
\caption{\label{fig:sim_fig1}Upper left $11 \times 10$ block of the
connectivity
matrix for the
simulated dataset.}
\end{figure}

We use the connectivity matrix to generate a connectivity graph. A connectivity graph has vertices that represent distinct organisms and edges that represent higher order terms between their reads. This definition will be made more precise in Section \ref{sec:analyses} when we deal with real data. The connectivity graph for the synthetic reads is in panel A of  Figure \ref{fig:sim_fig2}.  We verified that two genomes are connected in panel A if and only  if they are connected by reads in the connectivity matrix.

\begin{figure}
\centering
\centerline{
\includegraphics[width = 4in]{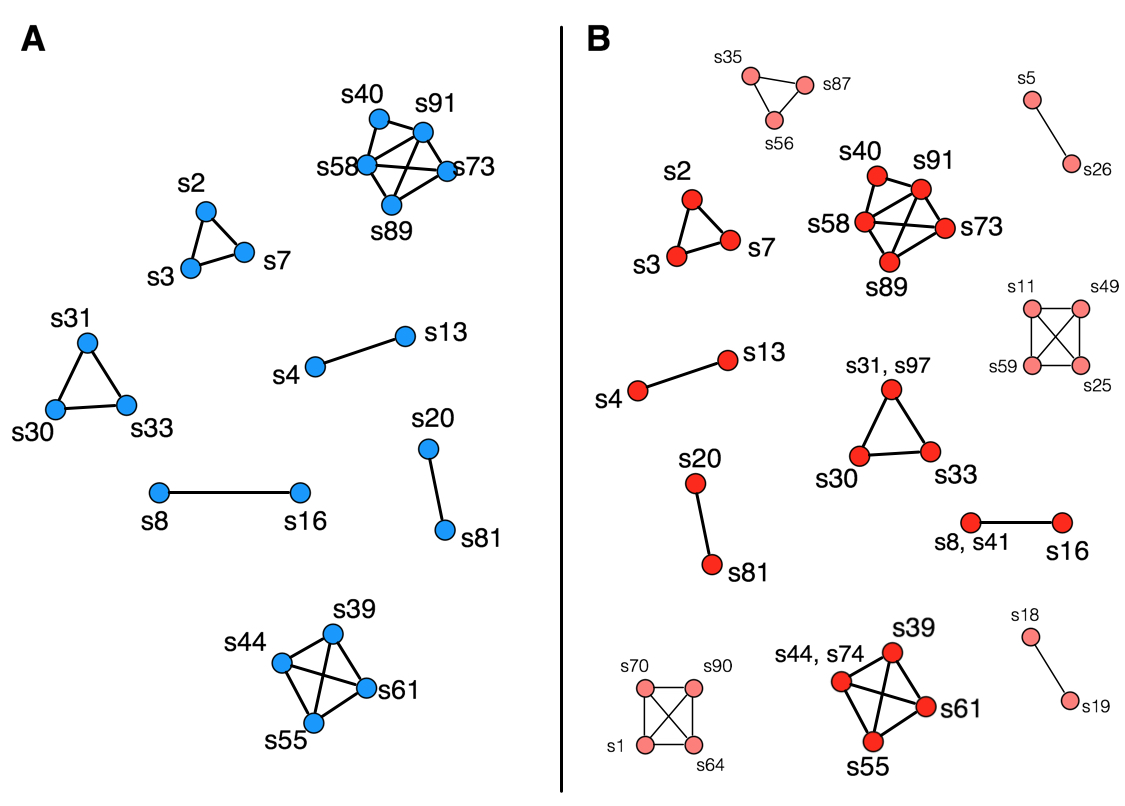}
}
\caption{\label{fig:sim_fig2}Panel A illustrates the known connectivities in the simulated data (blue). Panel 
B illustrates the connectivities generated by the model.}
\end{figure}

To verify our method generates a connectivity graph similar to panel A of Figure \ref{fig:sim_fig2},
we applied the method described in Section \ref{sec:ourmeth} using 200 chains each of length 200,000. We set $c = 10^{-4}$ in \eqref{eq:targetsetapprox}. We calculated BMA estimates of the posterior inclusion probabilities of edges in the corresponding independence graph based on the models in the set $\widehat{\mathcal{S}}(c)$. This generates the graph in panel B of Figure \ref{fig:sim_fig2}. We comment that there is nothing unique about the value $10^{-4}$:  it was chosen for convenience, was not discredited by the individual posterior probabilities we found, and a sensitivity analysis showed that it was a reasonable choice within a range of possible cutoff values. In practice, the choice of cutoff value would be data-driven to ensure appropriate robustness of the inferences.

In panel B of Figure \ref{fig:sim_fig2}, two genomes are connected if and if only if the sum of the posterior weights of the best models containing higher order terms between the two genomes is above $0.1$. Loosely, this is intuitively equivalent to saying that the posterior probability that the two genomes are associated (in the sense of higher order terms in loglinear models) is at least $0.1$. Comparing panels A and B, it is seen that every structure in A is in B although not every structure in B is in A. In panel B, the dark red circles indicate the pre-defined cliques built into the connectivity matrix and the faint red circles indicate extra cliques our method found by chance. So, in this simple case,  our method returns the full set of cliques built into the data. Thus, it is seen that even though clique loglinear models are a very restricted class, they can over-detect interactions. The reason is that in the simulated connectivity matrix, the cliques that are present are strongly built into the matrix; they will be found as long as enough reads are included.  However, cliques that are not present may also be found by our method from the random 1s in the connectivity matrix that do not correspond to genomes in any pre-defined cliques.  

The running time for our procedure -- using 200 genomes and 1000 reads -- was around six days. We used one compute node in a cluster; it had 16 CPU cores and 62 GB of memory. The code could be made significantly faster if implemented in C.  Our implementation in R is slower due to the limits imposed by this software package. Regardless, this simulated example goes well beyond earlier research with loglinear models in which 50 or fewer variables were used.   

\section{Example: characterizing associations in a microbial community}
\label{sec:analyses}

The Human Microbiome Project (HMP) is an ongoing collaborative study funded by the U.S. National Institutes of Health (NIH) to provide data and tools for studying the role of human microbiomes in human health and disease. Started in 2007 it has generated ground-breaking publications \cite{fierer-et-2010,zhao-et-2012,minot-et-2013}, and a plethora of metagenomic data on human microbiomes. Our method from Section \ref{sec:ourmeth} can represent the associations from an HMP sample with an independence graph
so we can infer the bacterial taxa present and their associations.

Human metagenome sample SRS015072, obtained from the vaginal microbiome of a female participant of the HMP Core Microbiome Sampling Protocol A (HMP-A) dbGaP study, was downloaded via FTP from the HMP Data Analysis and Coordination Center (DACC).  The sample consisted of $495,256$ paired-end, $100$ base pair reads (with an average mate-distance of $81$bp) sequenced and provided in Illumina FASTQ format. These reads were aligned to the collection of $4,940$ bacterial genomes, from the Integrated Microbial Genomes and Metagenomes (IMG, version 4.0) database \cite{markowitz-et-2014} using the Bowtie2 aligner \cite{langmead-salzberg-2012}. Of the sample reads, $369,633$ 
aligned to one or more of the reference genomes. The number of reads that aligned to each bacterial strain, species, and genera was calculated and connectivity tables of the form seen in Figures \ref{fig:contingency_table} or \ref{fig:sim_fig1} were generated for analysis at the genera level. 

The first step in each analysis is to identify those genera that cannot be involved in higher order interactions (i.e., cannot be part of a clique with two or more vertices). Note any two genera that define a marginal two-by-two table (disregarding all other genera) whose counts are not strictly positive cannot be part of the same clique because the MLEs of any clique loglinear model that involves that two-way interaction do not exist. We refine the definition of a connectivity graph as follows: it is a graph whose vertices correspond to categorical variables, i.e., the presence of a genus.  Given two vertices, there is an edge joining them if the two-way marginal contingency table associated with the two categorical variables contains only strictly positive counts. Within each analysis we ran the stochastic search from Section \ref{sec:ourmeth} for 100,000 iterations from 100 random starting clique graphs.

A total of $95$ genera had component species or strains to which reads aligned. Two genera were said to be connected if and only if each had at least one strain that shared a read. The genera with shared reads are shown in Figure \ref{fig:generaConnectivityGraph}. It is seen that 15 genera did not share reads across other genera (though each did within its own genus). This shows two facts:  i) 15 genera can be dropped from subsequent analysis at the genus level, and ii) the hairball showing the 80 genera that share at least one read is complex enough that further analysis is worth doing, i.e., it is worthwhile to use clique loglinear models to seek higher order interaction terms.

\begin{figure}
\centering
\centerline{
\includegraphics[width = 4in]{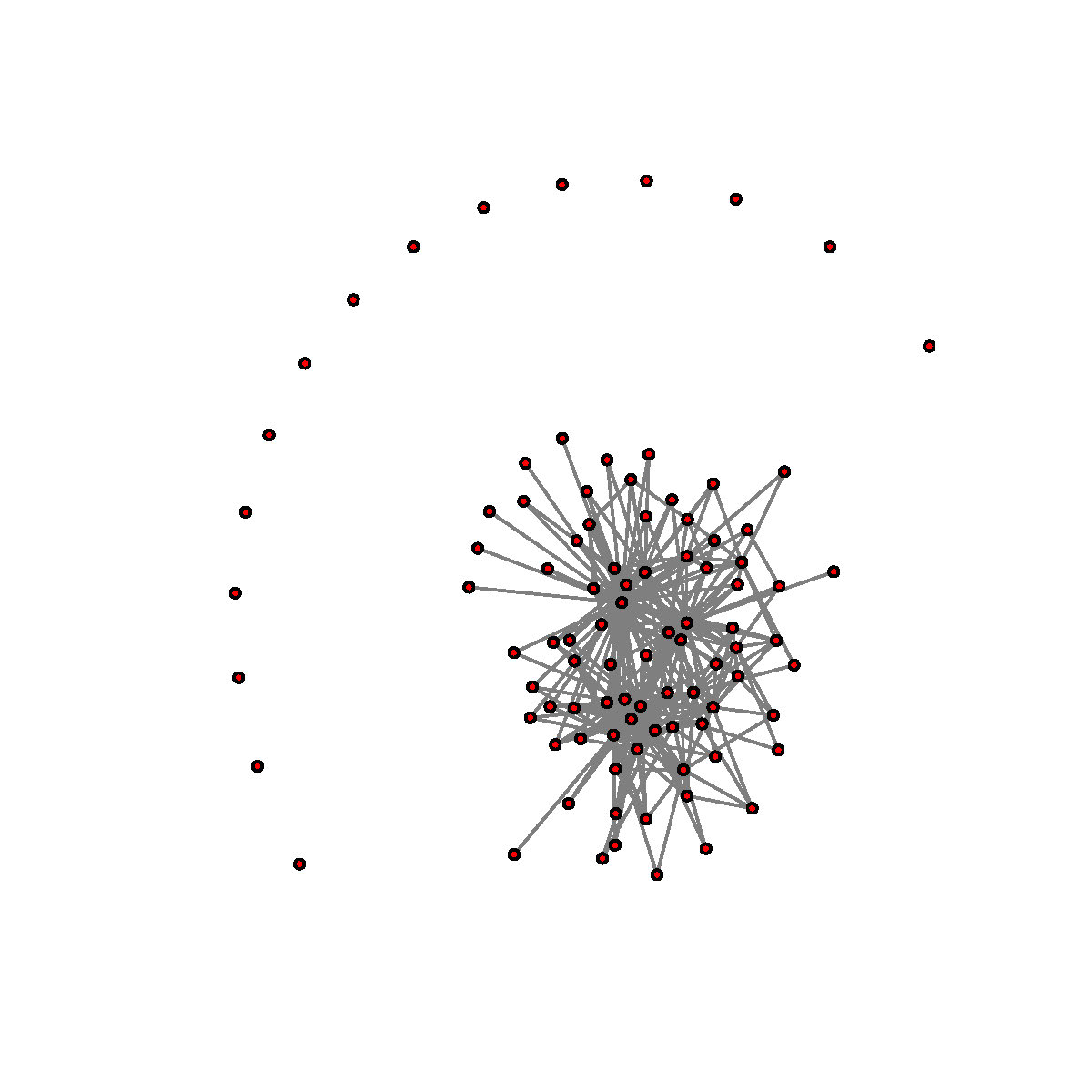}
}
\caption{\label{fig:generaConnectivityGraph}Genera connectivity graph for the example in Section \ref{sec:analyses}.}
\end{figure}

To get an idea of the level of complexity of the hairball in the raw connectivity graph in Figure \ref{fig:generaConnectivityGraph}, we generated Figure \ref{fig:generaConnectivityNeighbors}. It is a bar graph that gives the degree (the number of neighbors) of each element (genus) in the raw connectivity graph as a 
proportion of the total number of vertices (or genera). Several of these genera, including Lactobacillus, Prevotella, and Staphylococcus, have been identified as common members of vaginal microbiome communities \cite{huang-et-2014}. 

\begin{figure}
\centering
\centerline{
\includegraphics[height = 4.25in]{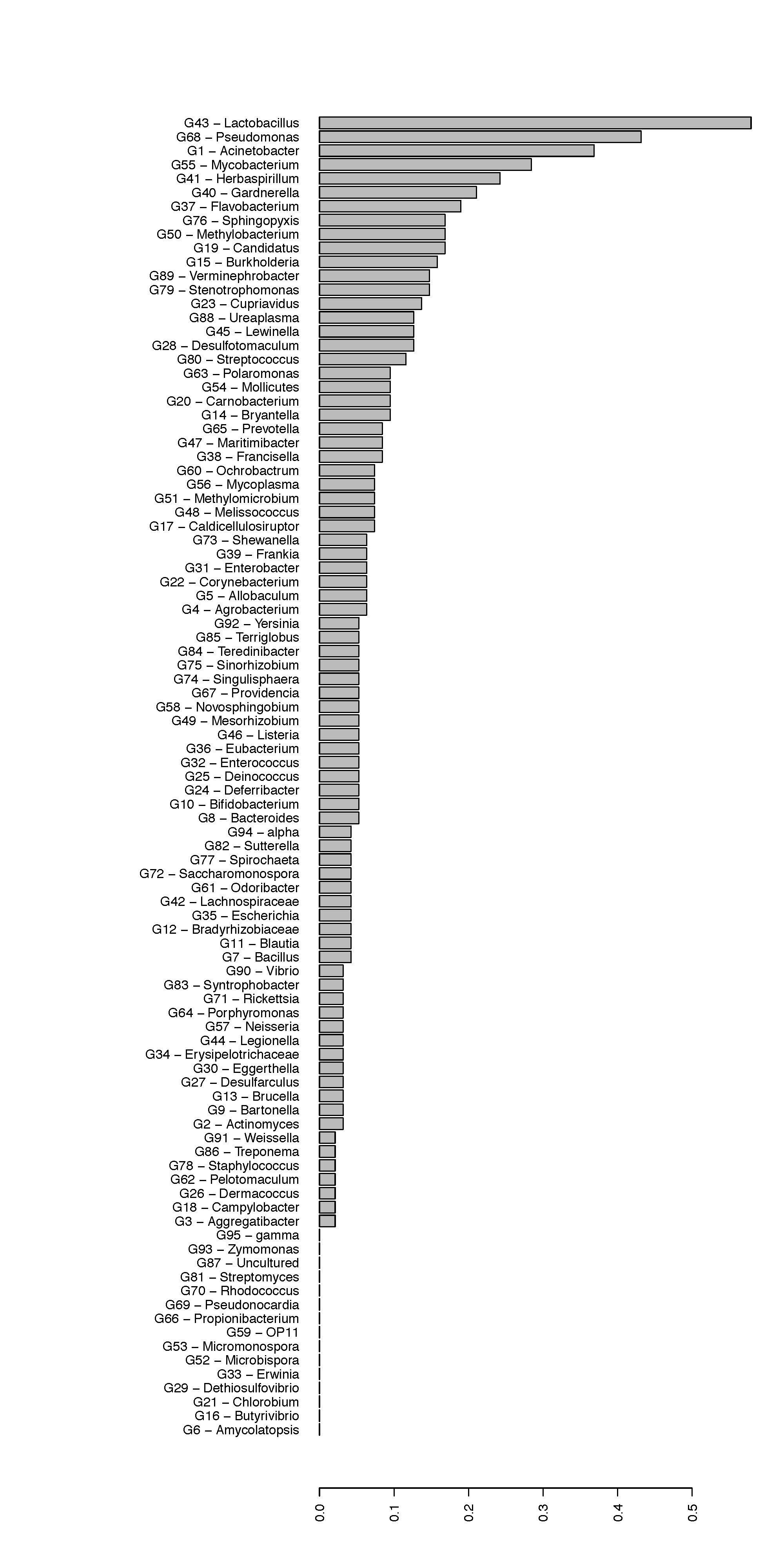}
}
\caption{\label{fig:generaConnectivityNeighbors}Degree of each vertex in the raw 
connectivity graph in Figure \ref{fig:generaConnectivityGraph}, expressed as a proportion of the total number of vertices, for the genera data.}
\end{figure}

Note that the 15 isolated genera in Figure \ref{fig:generaConnectivityGraph} cannot form cliques with any of the other genera because in the reduced table, i.e., grouping all strains into genera, the two-way marginals they form contain two counts of zero.  So, they cannot be accommodated by clique loglinear models -- except as cliques of size one -- and dropping them amounts to a significant reduction in computational running time. This is important because the number of possible clique loglinear models increases rapidly with the number of vertices.

\begin{figure}
\centering
\centerline{
\includegraphics[width = 4in]{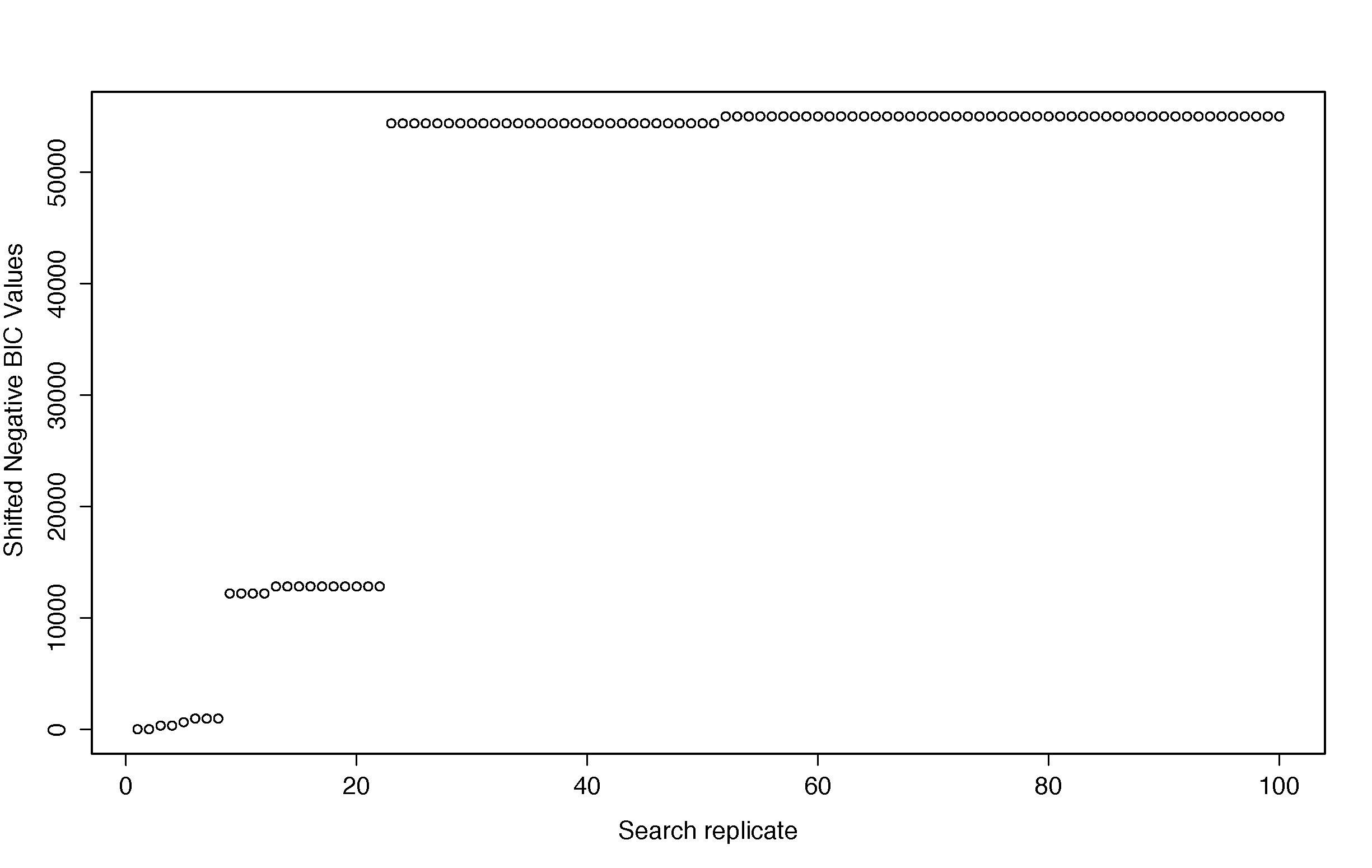}
}
\caption{\label{fig:generaconvergence}
For each run of the stochastic search algorithm, we recorded the smallest value of the BIC it identified. Smaller values (on the vertical axis) indicate better models and the eight models in a row at the bottom left have roughly the same worth in a BIC sense.  As shown in the Supplementary Materials Section \ref{sec:robsparse}, as sparsity decreases, we obtain smoother graphs than the step-function appearance seen here.}
\end{figure}

Now we have reduced the data to an 80-dimensional contingency table.  It has 377 cells with strictly positive counts. The largest positive count is 332,117 while the second largest is 11,614. Running the model selection procedure from Section \ref{sec:ourmeth} generates a series of clique loglinear models and the BIC value of each can be calculated. Figure \ref{fig:generaconvergence} shows the BIC values associated with the 100 best clique loglinear models identified by each of the 100 runs of length 100,000.  Some chains were trapped in local modes, but others found their way towards what appears be the best clique loglinear model in terms of lowest BIC. Of all the clique loglinear models visited, we found 1133 whose BICs were within $c=10^{-4}$ of the BIC of the best clique loglinear model identified across all 100 chains. As in Section \ref{sec:sims}, we used an Occam's window form of BMA limited to the best models visited while renormalizing \eqref{eq:targetdist} to reflect this. The clique structure of the best model is shown in Figure \ref{fig:bestgraphGenera}.

\begin{figure}
\centering
\centerline{
\includegraphics[width = 2.5in]{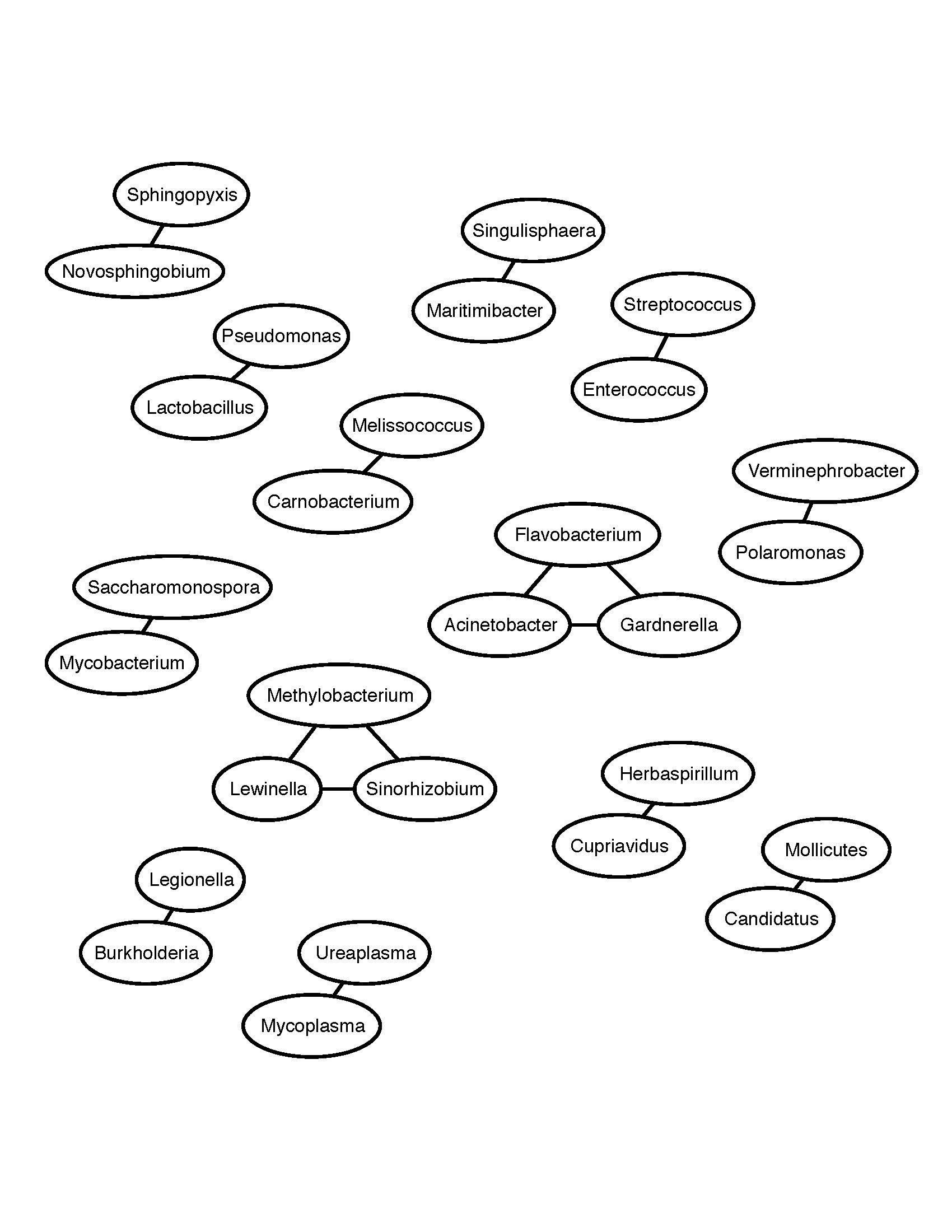}
}
\caption{\label{fig:bestgraphGenera}The best clique loglinear model identified for the genera data. Genera names are consistent with those in Figure \ref{fig:generaConnectivityNeighbors}.}
\end{figure}

We have also generated the independence graph resulting from the Occam's window BMA in Figure \ref{fig:bmagraphGenera}. In this graph, the strength of the connectivities between genera are indicated by different colors of lines.  These reflect ranges of posterior probabilities calculated from the Occam's window BMA probability using the models amongst the 1133 best models for which a given collection of higher order terms is present.  Observe that Figure \ref{fig:bmagraphGenera} does not have a clique structure because BMAs of clique loglinear models do not in general form another clique loglinear model.  More specifically, in both Figures  \ref{fig:bestgraphGenera} and \ref{fig:bmagraphGenera}, the edges correspond to pairs  of variables/genera that have strictly posterior probabilities of belonging to a collection of higher order terms in the 1133 clique loglinear models as evaluated by the posterior probabilities given by the Occam's window BMA. 

Several of the links in Figure \ref{fig:bmagraphGenera} are supported by biological findings regarding the vaginal  microbiome.  For instance, i) Ureaplasma and Mycoplasma are bacterial genera from the same  bacterial family and are commonly found in the reproductive tract of both men and women, ii) Polaromonas and Verminephrobacter, and Yersinia and Caldicellulosiruptor, are from the same bacterial family and have been validated by experimentation, and iii) Melissococcus and Carnobacterium are both main genera producers of bacteriocins, ribosomally synthesized antibacterial peptides/proteins that either kill or inhibit the growth of closely related bacteria and are considered antimicrobial microbes.

\begin{figure}
\centering
\centerline{
\includegraphics[width = 2in]{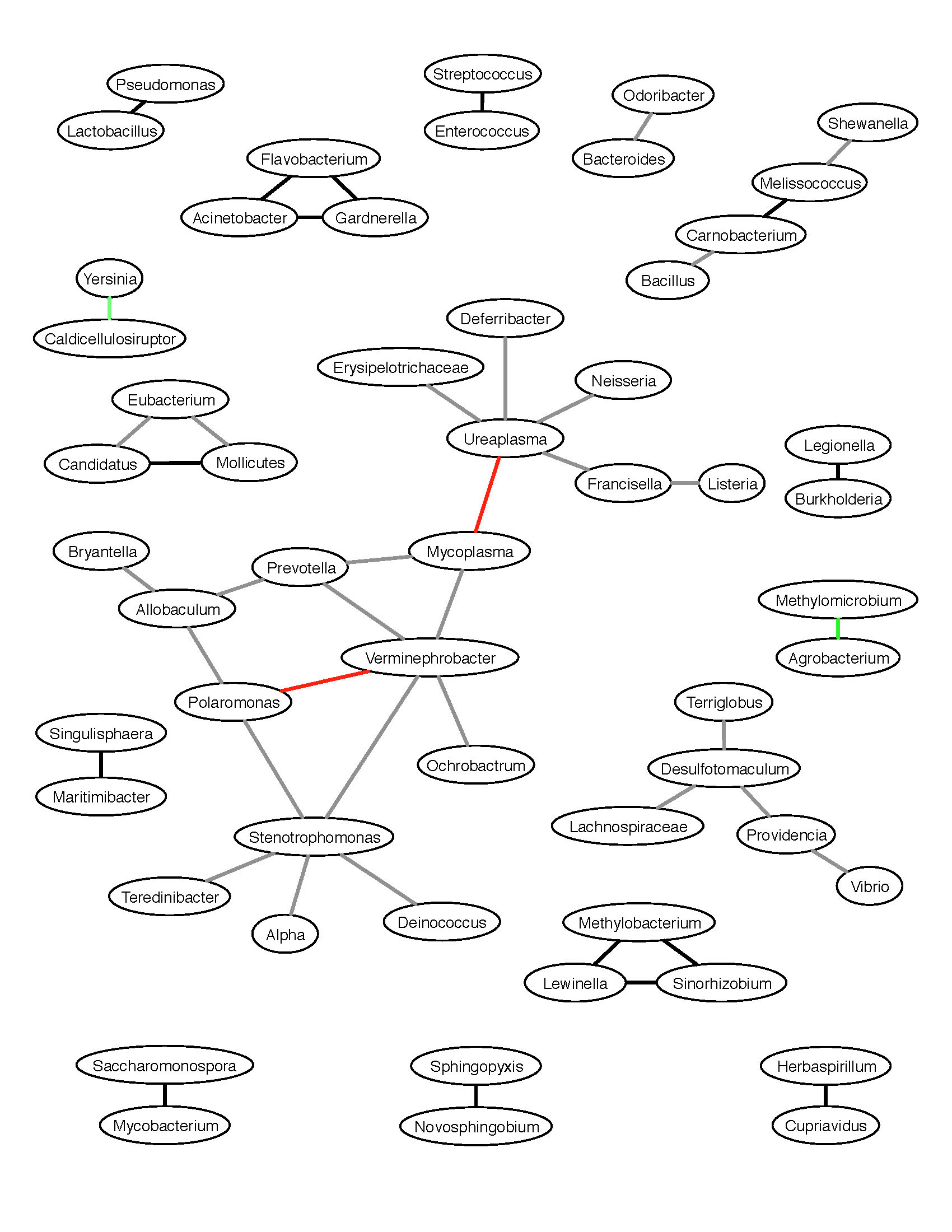}
}
\caption{\label{fig:bmagraphGenera}Pairs of variables that have strictly positive posterior probabilities of belonging to an higher order term of a clique loglinear model in the genera data under the
Occam's window BMA.   Black, red, green, and grey lines 
denote posterior probabilities that belong to the intervals $(0.9,1]$, $(0.5,0.9]$, $(0.1, 0.5]$, 
and $(0, 0.1]$, respectively. Genera names are consistent with those in Figure \ref{fig:generaConnectivityNeighbors}.}
\end{figure}

We also produced estimates of the probabilities in \eqref{eq:probunknown} and \eqref{eq:probshortsequence} which, from here on, will be referred to as individual existence probabilities. The BMA estimates of the top 5 individual existence probabilities are as follows:  Lactobacillus 0.86, ``Unknown'' 0.08, Pseudomonas 0.05, Acinobacter 0.01, and Gardnerella 0.002 (probabilities do not add to one because of rounding).  While we do not have standard errors for these estimates, it is obvious that Gardnerella is present in only trace amounts and the presence of an unidentified genus is not zero. We return to the interpretation of unidentified genera in Section \ref{sec:discussion}, but note that our findings are consistent with those from analyses of vaginal microbiome samples based on 16S rDNA sequence data that also identified the presence of previously unknown bacterial taxa \cite{fettweis-et-2012}.  

\section{Example: the diabetic foot wound microbiome}
\label{sec:footwound}

One of the complications of diabetes, particularly in elderly patients, is the development and impaired healing of foot ulcers. Diabetes is the primary cause of non-traumatic lower extremity amputations in the United States; approximately 14-24\% of patients with diabetes who develop a foot ulcer eventually require an amputation. A diabetic foot ulcer is an open sore or wound that occurs in about 15\% of patients with diabetes, usually on the bottom of the foot. 

It has been hypothesized that an altered skin microbiome may play a role in the compromised healing of diabetic foot ulcers \cite{smith-et-2016}. The purpose of this study was to investigate  the bacteria involved in chronic wound healing. Samples were taken from three locations -- the wound bed, the wound edge, and the peripheral healthy skin of the foot -- of 10 patients at two time  points, the time of initial visit and one week after the initial visit.  Half of the patients were considered healers, and the remaining patients were  considered nonhealers based on clinical assessment of their wounds.  Samples were prepared and submitted for V4 16S rRNA gene sequencing with the Illumina MiSeq platform, with services provided by Second Genome, Inc. Of the original samples, a  total of 50 provided reliable sequencing information. Sequences from these samples  that passed quality filtering were mapped to a set of representative consensus sequences  to generate an abundance table of Operational Taxonomic Units (OTUs); an OTU is simply a cluster of closely related reads. This table was  analyzed using an overdispersed Poisson model \cite{robinson-et-2010} to identify OTUs that were significantly differentially expressed between healers and  nonhealers at each of the three sample locations (FDR corrected p-value $<$ 0.05). The results consisted of three lists of significant OTUs, one for each location.  Further details on this can be found in the Supplementary Materials, Section \ref{sec:wound-supp}, that also extend our method to the comparison of two populations.

Our clique loglinear analysis is based on the counts of the number of sequencing reads assigned to significant OTUs for each sample, with a separate table for the significant OTUs from each location. This analysis is different from our previous example in that: i) our interest is on the association among samples that may be reflected in components of the microbiome, and ii) any associations will be based on sharing of OTUs across samples as opposed to sharing of reads across genomes. 

An initial exploratory data analysis using hierarchical clustering and principal components  analysis revealed that samples from the same subject cluster together, and that subjects cluster into two groups with patients 4 and 5 (healers) and patient 7 (nonhealer) forming a cluster distinct from the remaining subjects (see Appendix, Section \ref{sec:wound-supp}). For each of our three analyses (one for each significant OTU list) all samples could  form cliques with any of the other samples because all two-way marginals contain only 
nonzero counts.  We ran the stochastic search from Section \ref{sec:ourmeth} for 100,000 iterations from 100 random starts. Due to the smaller size of the table relative to the previous example (50 vs. 80 binary variables) we noted convergence to a best graph in less than 50,000 iterations. 

The strongest factor in clique formation is subject/patient origin followed by sample location; the distinction between healers and nonhealers is not evident despite the focus on significant OTUs. 

\begin{figure}
\centering
\centerline{
\includegraphics[width = 3.75in]{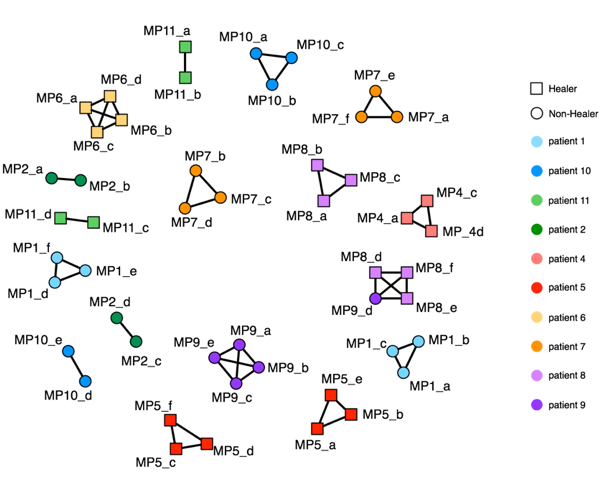}
}
\caption{\label{fig:bestGraphBed}The best clique loglinear model of samples from the wound microbiome based on significant OTUs in the wound bed between healers and nonhealers. MP = microbiome patient. The sample location is wound bed (a and b), wound edge (c and d), or healthy skin (e and f).}
\end{figure}

Although some cliques appear in all three best graphs, e.g., samples from the wound edge of patient 2, most cliques shift subtly with the changes in the significant OTUs; see 
Figure \ref{fig:bestGraphBed}. For example, all samples from patient 6 form a single clique in the best graph based on OTUs that are significantly different in the wound bed between healers and nonhealers. However, in the best graph based on significant OTUs in the wound edge, one of the wound bed samples from patient 6 forms a clique with the wound bed samples from patient 11, while the remaining samples from patient 6 maintain a clique. Not surprisingly, the cliques involving samples from patient 6 change again in the best graph based on significant OTUs in healthy samples (note that patient 6 has no healthy samples). As expected from our exploratory analysis, samples from patients 4, 5, and 7 formed cliques only among themselves and not with samples from other patients. 

Thus, we have analyzed the data to identify the OTUs and to search for associations among samples. To the best of our knowledge, the data in this example have not been analyzed in any way analogous to our methodology before, either  biologically or statistically. Thus, we are unable to corroborate our findings from other sources although a priori our findings do not appear unreasonable. This example demonstrates that our methodology has the potential to obviate a lot of expensive microbiological work.

\section{Discussion}
\label{sec:discussion}

We have developed a statistical methodology for contingency table analysis based on clique loglinear models. This methodology can be used in the context of high dimensional tables and it accounts for model uncertainty by model averaging. Our methodology can infer the presence/absence of specific taxa as well as associations among members of taxa. We have demonstrated our approach in both simulated and real data contexts relevant to applications to metagenomics.

An issue that repeatedly comes up in this sort of analysis is how to account for the dependence among reads.  As noted in Section \ref{sec:intro}, this dependence is frequently small -- and in our experience, the higher the quality of the data the smaller the dependence among reads is, though it is never zero. As a first approximation, therefore, assuming independence is reasonable and parallels the the bag-of-words approach that has been applied with much success in natural language processing.  In recent years this has been improved to random `$N$-grams' and an analogous improvement may be possible with NGS reads. Of more immediate relevance, the-bag-of-words model has been applied to several branches of bioinformatics with success \cite{lovato-2015}. It is important to distinguish between genuine associations among taxa and simply having reads in common for some other reason  -- syntrophy or evolution for example -- something our methodology does not address. However, this level of study remains in its infancy. 

An important question is to ask how much information is really in the reads.  In this context, one can ask if there is adequate read converge to infer reliably which genomes are present in the sample and hence in the population.  Obviously, this is a function of the number of reads, diversity within the sample, and complexity of the population. This is not a question that can be addressed statistically after the reads have been generated, although in principle it could be partially addressed at the design stage of the read generation. Read coverage will typically be incomplete and typically will be a limitation on analytic methods. This may increase the uncertainty of downstream inferences but the task is to reflect uncertainty accurately not to under-represent it.  Methodologies that compensate for uncertainty or evaluate uncertainty by, say, robustness criteria, remain to be developed.

More specifically, when reads are shared by two taxa they only mean that the two taxa are similar in the regions that were sampled.  Strictly speaking this does not tell us anything about the co-existence of the two taxa.  However, first, if shared reads from two taxa are found we have ruled out the case that neither taxon is present. Moreover, if we have reads present that are unique to one taxon then we have established its presence.  If we have reads that are unique to the other taxon then we have established its presence.  Finally, if we have reads that are unique to each of the taxa we have unambiguously identified both taxa are present.  We dropped reads that are unique to one taxon they were the singleton vertices in the connectivity graphsâ so we could focus on higher order terms that represented two reads. 

We comment that in our wound microbiome example, unlike the HMP example, our analysis may have failed to capture all of the information in the available data as the OTU table was converted from counts (i.e., each cell gives the {\it number} of reads that align to a given OTU in a given sample) to binary (i.e., each cell indicates if {\it any} reads align to a given OTU in a given sample) prior to analysis. An extension of our method to raw data consisting of counts, e.g., with each subject $S_{j}$, $1\le j\le b$, we associate a categorical random variable $X_{j}$ that takes value $k$ if $k$ sampled reads align to OTU $i$, $1 \le i \le k$ and takes value $0$ otherwise, is a topic for future research.   

On the other hand our method extends to comparing two populations on the basis of their connectivity. The difference in dependencies can be regarded as indicators of which associations are present or absent in the normal case (say) versus the diseased case.  This amounts to looking at the different structure of the graphs and interpreting what the cliques mean in terms of reads. Our treatment in Section \ref{sec:footwound} was subject-by-subject. In Section \ref{sec:wound-supp} of the Appendix we compare two collections of metagenomic samples from two populations, healers and nonhealers.  

Our inference of the significant presence/absence of bacterial taxa, possibly unknown, is based on posterior estimates of probabilities \eqref{eq:probunknown} and \eqref{eq:probshortsequence}. We refer to these as existence probabilities; however, this terminology belies the subtleties regarding their interpretation. These probabilities are estimated from the model averaged joint posterior distribution, and hence are conditional on the best BIC models, i.e., an Occam's window approach. If a bacterial taxon (say, genus) to which reads uniquely align does not appear in any of these models, these reads will impact our probability estimates. For example, the estimate of the probability of an unidentified taxon will be inflated, while the estimate of the probability of presence of a taxon appearing in the model average will be deflated. The extent of this impact may be small, as any taxon to which many reads align should appear in the model average, but this cannot be guaranteed and warrants further study.  

Otherwise put, the category of genomes we have called unidentified may only be an artifact of the modeling.  Indeed, if the genome list contains all the genomes in the sample, the probability of an unidentified genome is simply a residual reflecting the short reads that do not align in sufficiently large numbers to any genome in the models in the model average.  On the other hand, if the genome list is incomplete, the probability of an unidentified genome is the sum of two parts:  the probability of a genome we know but that was not included amongst the $B$ reference genomes plus the probability of something that we have not encountered before.

\section*{Acknowledgment}

This work was supported in part by the National Science Foundation through grants DMS/MPS-1737746 and DMS-1120255 to University of Washington, and grants DMS-1410771 and DMS-1419754 to University of Nebraska-Lincoln. 

\appendix

\section{Data preparation}
\label{sec:data}

The bacterial genomes data used in our experiments were obtained from the Human Microbiome Project (HMP) \cite{nih-et-2009}, and the Joint Genome Institute's (JGI) Integrated Microbial Genomes (IMG) database \cite{markowitz-et-2014}, version 4.0.

\subsection{Bacterial genomes}
A number of 456,865 genomic bacterial reference sequences, in FASTA format, were procured from the Integrated Microbial Genomes and Metagenomes (IMG, version 4.0) database \cite{markowitz-et-2014}.  The 456,865 reference sequences accounted for 5,168 bacterial genomes (at strain level) which included sequences from bacterial genomes and bacterial plasmids.  The 5,168 genomic references were isolated by relying on bacterial taxon names and identifiers obtained from the Genome Browser at the IMG website:

{\tiny \url{https://img.jgi.doe.gov/cgi-bin/w/main.cgi?\\section=TaxonList&page=taxonListAlpha&domain=Bacteria}}

\subsection{Metagenomic samples}
A human metagenomic sample was obtained from the Human Microbiome Project \cite{nih-et-2009}.  The human metagenome sample SRS015072, obtained from a female participant of the HMP Core Microbiome Sampling Protocol A (HMP-A) dbGaP study, was downloaded from the HMP FTP site.  The sample consisted of  495,256 paired-end, 100 bp reads (with an average mate-distance of 81bp) in Illumina FASTQ format \cite{cock-et-2010}.

\subsection{Reference sequence alignment}
The HMP dataset (SRS015072) was aligned to the 456,865 bacterial references using the Bowtie2 \cite{langmead-salzberg-2012} aligner.  Bowtie2 requires that the reference sequences be indexed so that the reads can be efficiently aligned. The bacterial genomic references were prepared for alignment using the \texttt{bowtie2-build} indexer program.  The indexer program was run with default values along with the ``-f'' flag (FASTA sequences). Once the reference sequence index had been built, Bowtie2 used in local-alignment mode, was used to align the HMP data using the following command:

{\small
\begin{verbatim}
bowtie2 --local -D 20 -R 3 -N 0 -L 20 -i S,1,0.50 --time -f -x -S
\end{verbatim}
}

Samtools \cite{li-et-2009} (version 0.1.18) was then used to parse the alignments.

\subsection{Whole genome sequencing vs. 16S sequencing}
\label{sec:sequencing}

In the example from Section \ref{sec:analyses}, the whole genome sequencing (WGS) data were used, while in the example from Section \ref{sec:footwound}, the 15S sequencing data were used.  Both datasets can be seen as special cases of the whole metagenome sequencing data our methodology focuses on \--- see Section \ref{sec:intro}.

In WGS, sequencing libraries are prepared from the extracted whole-DNA sequences of bacteria in the sample to be analyzed. The resulting sequencing short-reads consequently represent the putative DNA sequences of the bacterial populations in the sample \cite{thoendel-et-2016,hasman-et-2014}. In contrast, 16S rRNA sequencing libraries are prepared from the sequences of the highly conserved 16S ribosomal gene. The reads from 16S sequencing represent the sequences of the 16S gene in the bacterial populations in the sample. In downstream analyses, WGS data are analyzed at the sequencing-read and genome levels, while 16S reads are assembled into clusters of reads called Operational Taxonomic Units (OTU), and analyzed as abstract representations of taxonomic groups \cite{nguyen-et-2016,charuvaka-rangwala-2011}. WGS analyses have the advantage that their analysis resolutions tend to be higher than 16S when detecting bacterial genomes in a sample they can detect bacterial organisms at low taxonomic levels.  WGS capture a larger region of a bacterium's genome than 16S, and are able to detect more given the appropriate depth and breadth of sequencing coverage, e.g. how much of a given bacterium's genome is covered by the sequencing library (breadth), and how many sequences sample a given region (depth) \cite{ranjan-et-2016}. The 16S analysis is a well-known procedure with mature and well-developed analysis pipelines, as well as a large collection of publicly available datasets. The 16S sequencing can also achieve the same levels of classification as WGS, but at a  cheaper cost and with smaller sequencing libraries which can have downstream benefits as their  data footprints are more compact, resulting in cost effective, well-understood analysis pipelines.

\section{Some methodological and computational details}

There are a variety of methodological and computational questions that arise in complex, big data analyses such as the ones we present in this article. This section is intended to address several of the most important ones.

\subsection{From loglinear models to clique loglinear models}
\label{sec:nested}

A loglinear model arises from using an ANOVA model for the log expected frequencies in a contingency table.  As is well known, the number of terms in a complete 
ANOVA model is exponential in the number of factors.  So, to be useful, we must have a way to reduce the number of models considered.  One way to do this is to restrict the class of loglinear models. The nested subclasses are shown in Figure \ref{fig:loglin}. The first restriction we impose is that the loglinear models be hierarchical. A loglinear model is hierarchical if and only if the presence of a higher-order interaction term requires the presence of any or all of its lower-order interaction terms \cite{bishop-et-2007}.  The model may not be identifiable without imposing some constraints on the interaction terms \cite{agresti-1990}. 

\begin{figure}
\centering
\centerline{
\includegraphics[width = 3in]{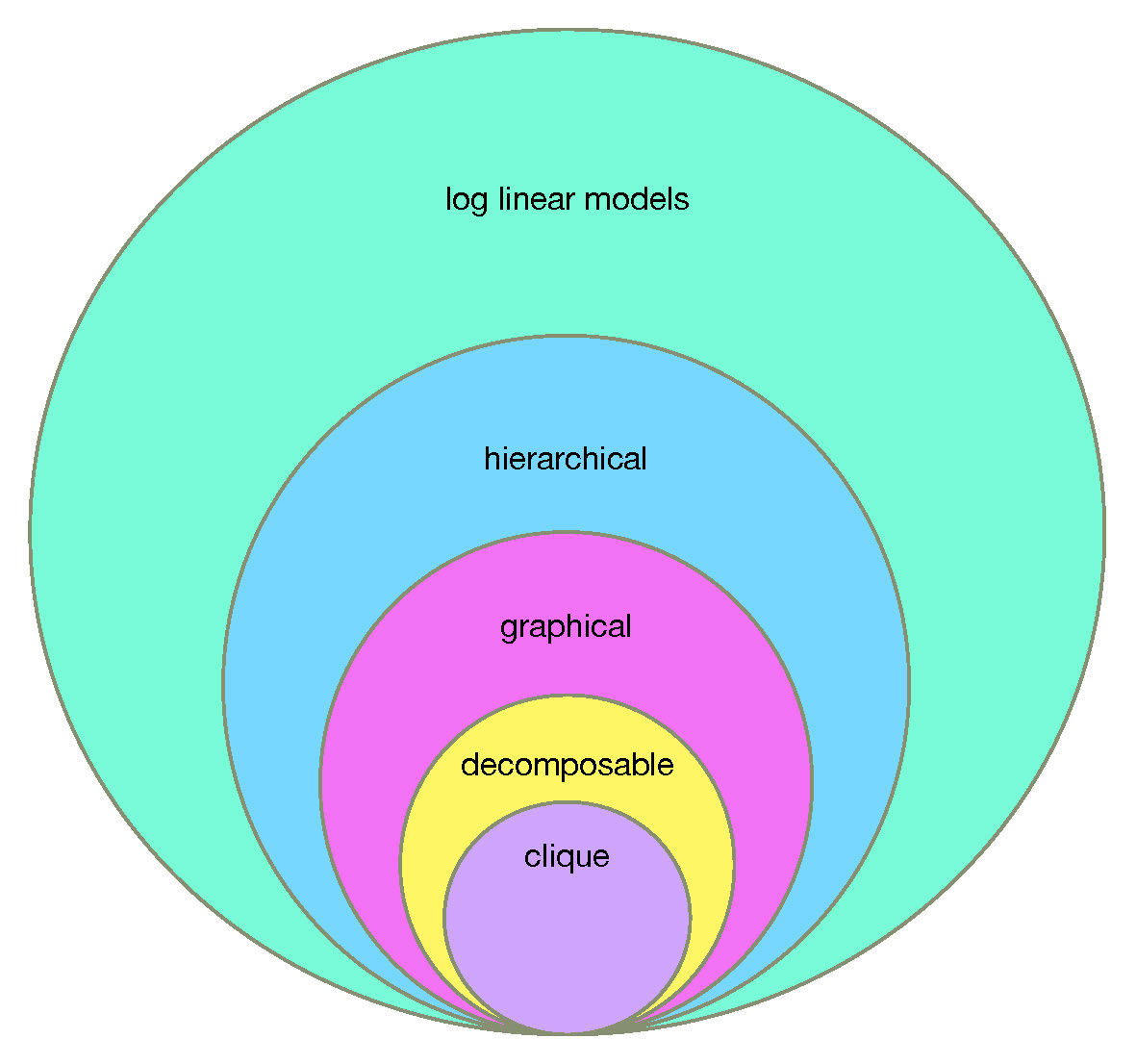}
}
\caption{\label{fig:loglin}Nested subclasses of loglinear models.  We propose using the smallest of these -- clique loglinear models --  for sparse high-dimensional contingency tables.}
\end{figure} 

The second restriction is to hierarchical loglinear models that are also graphical. The idea is that an undirected graph $G$, with vertices $\mathcal{B}$ and edges $E$, can be associated to any hierarchical loglinear model. These graphs are usually called independence graphs \cite{whittaker-1990}. Specifically, an edge $e=(b_{1},b_{2})$ appears in $G$ if and only if the variables $X_{b_{1}}$ and $X_{b_{2}}$ appear together in an interaction term.  A hierarchical loglinear model is graphical if and only if the subsets of $\mathcal{B}$ that are the vertices of the complete subgraphs of $G$ that are maximal with respect to inclusion, are also maximal interaction terms in the loglinear model \cite{whittaker-1990}.  If a model is graphical, the absence of an edge $e=(b_{1},b_{2})$ in $G$ is  equivalent with the conditional independence of variables $X_{b_{1}}$ and $X_{b_{2}}$ given the rest of the variables under the joint distribution for $X_{\mathcal{B}}$.  Moreover, if there is no path in $G$ from vertex $b_{1}$ to vertex $b_{2}$, then $b_{1}$ and $b_{2}$ are in two distinct and fully connected components of $G$, and the corresponding variables $X_{b_{1}}$ and $X_{b_{2}}$ are independent. 

The usual notation for this is to abbreviate the components in $\mathbf{X}$ to their indices $b\in \mathcal{B}$,
and indicate loglinear models in terms of sets of indices that show the clique structure of $G$.   For example, consider the loglinear models $\M_{1}$ and $\M_{2}$ indicated
by their generators, i.e., generated by the noted collections of sets of indices, ${\mathcal{C}}(\M_{1}) = \{ \{1,2,3\},\{1,3,4\}\}$  and ${\mathcal{C}}(\M_{2}) = \{ \{1,2\},\{2,3\},\{1,3\},\{3,4\},\{1,4\}\}$. The graphs associated with $\M_{1}$ and $\M_{2}$  are identical, and both have edges as given in ${\mathcal{C}}(\M_{2})$. However, the graph $G$ for $\M_{1}$ is generated by two cliques $\{1,2,3\}$ and $\{1,3,4\}$  that include the third order interaction terms for $(1,2,3)$ and  $(1,3,4)$, whereas the list of edges $\M_{2}$ forms a graph $G^\prime$ that does not include the third order interaction terms for $\{1,2,3\}$ and $\{1,3,4\}$. Indeed, even if the edges are regarded as cliques of size two, the third order interaction terms necessary for a graphical model as indicated by the cliques in its graph are not included. Thus $\M_{1}$ is a graphical model, but $\M_{2}$ is not a graphical model.

\begin{figure}
\centering
\centerline{
\includegraphics[width = 3in]{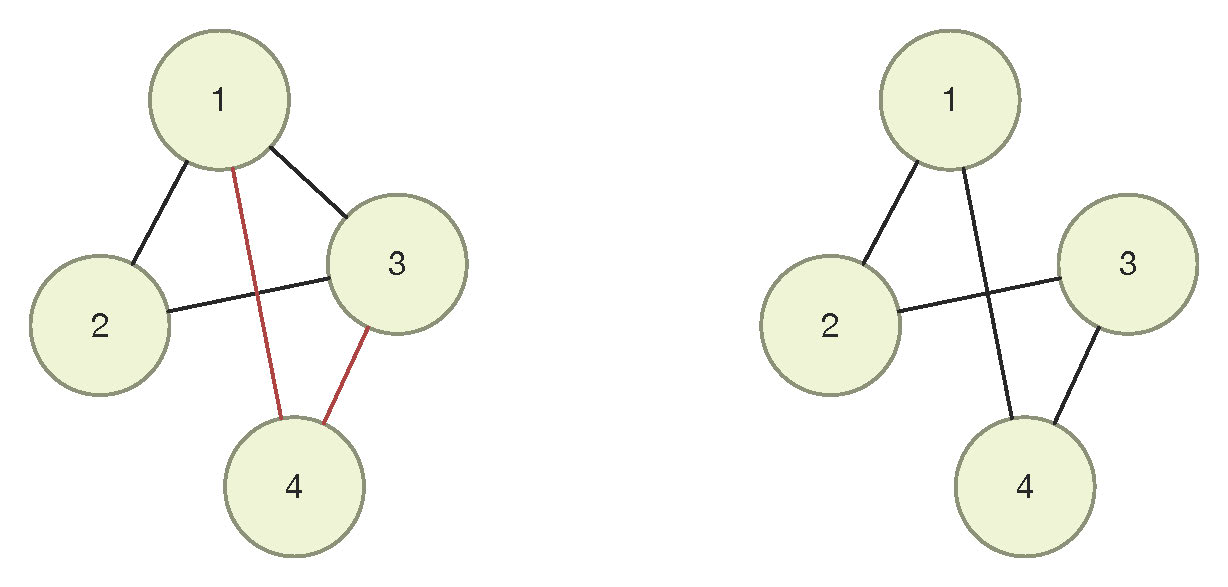}
}
\caption{\label{fig:graph_mod}Graphical loglinear models $\M_{1}$ (left) and $\M_{3}$ (right). $\M_{1}$ is decomposable while $\M_{3}$ is not.}
\end{figure}

The third restriction is to graphical loglinear models that are decomposable \cite{lauritzen-1996}. This idea is that a graphical model is decomposable if its
graph can be broken into components that are cliques without losing any information. 
For example, Figure \ref{fig:graph_mod} shows the graphical loglinear model
$\M_{1}$ and a second model $\M_{3}$. 
$\M_{1}$ is generated by ${\mathcal{C}}(\M_{1}) = \{ \{1,2,3\},\{1,3,4\}\}$ and is decomposable; its  
graph has two cliques $\{1,2,3\}$ and $\{1,3,4\}$, and one separator $S_{2}=\{1,3\}$. Separators are the intersections of cliques and have important properties not discussed here, but see \cite{lauritzen-1996}. By contrast, $\M_{3}$ is generated by ${\mathcal{C}}(\M_{3}) = \{ \{1,2\},\{2,3\},\{3,4\},\{1,4\}\}$ and is graphical.  However, it is not decomposable.  Essentially, decomposable models are graphical models for which closed form MLEs exist.  In terms of their graphs, this means they are triangulated, that is their chordless cycles contain no more than three vertices. The equivalence of these two formulations is beyond the scope of this brief review; see \cite{lauritzen-1996} for details.  One of the key advantages of decomposable graphical models is that testing the existence of MLEs is straightforward and computationally less expensive even when $B$ is large.

\subsection{Is this a true Bayes approach?}
\label{sec:reallybayes}

The analysis here is not obviously fully Bayes.  The lack of Bayesian-ness, however, is
in the Occam's windowing of the posterior  (which is data dependent) not in the prior selection.
Even though not strictly Bayesian, Occam's window is a standard technique in cases such as ours where the posterior is difficult to explore.  It allows us to focus on the most important regions of the model space. The argument for using is purely empirical \cite{madigan-raftery-1994}, and the technique was developed in response to loglinear models.  Three more points are relevant.

First, BIC with MLEs is actually equivalent to using the simplest of priors (uniform) for which 
the posterior mode (optimal under zero-one loss) equals the MLE.  Other priors could have 
been used -- at the cost of making an already computationally demanding approach more so.  
For instance, there are priors that assign weights to models in ways related to their number of terms, usually larger weights on smaller models.  This would have necessitated a full blown  stochastic search over the model space and massively increased running time.  Of course, unless we have lots of pre-experimental information we are sure is correct we run the risk of ending up with inappropriately prior-driven results since our model space, although reduced, is so large.  We have defaulted to the uniform prior for convenience, hoping that the asymptotic results discussed below will suffice.

Second, outside the uniform prior, it is well known that many Bayes estimators
and MLEs are very closely related.  For instance, if $\hat{\theta}$ is an MLE and $\hat{\theta}_p$ is the mean of the posterior, it is not hard to prove that $\sqrt{n}(\hat{\theta} - \hat{\theta}_p) \rightarrow 0$ in distribution even as $\sqrt{n}(\hat{\theta} - \theta_0)$ and  $\sqrt{n}(\hat{\theta} - \theta_0) $ go to a normal distribution. That is, the MLE and posterior mean are closer to each other than either is to the true value $\theta_0$.

Third, asymptotically approximating the posterior using the BIC is a very well-established 
technique. The most recent important reference on it seems to be \cite{berger-et-2003}, and they give some history of that type of approximation.  Since the prior only contributes to the second order approximations to the posterior (say ${\mathcal{O}}(1/n)$) using BIC with MLEs eliminates the need for prior specification when $n$ is large enough in well behaved cases, e.g., ignoring model misspecification issues (but see \cite{berk-1967} for the standard way such issues are still handled in model selection), non-identifiability or dilution \cite{george-2010}.

\subsection{Acceptance rates in an Metropolis-Hastings (MH) procedure}
\label{sec:whatabout}

Our stochastic search uses the mechanism of Metropolis-Hastings (MH) \cite{madigan-york-1995}. The MH acceptance rates will primarily be a function of the richness of the model space -- which is largely a subjective choice -- not an indicator of whether a satisfactory collection of good models has been found although the two may be related. Indeed, there is a bias-variance tradeoff on the level of model lists:  too big a model list can give excessive variance; to small a model list can give bias. Finding a model list that achieves an optimal variance bias tradeoff is very much an open question.

\begin{figure}
\centering
\centerline{
\includegraphics[width = 5in, height= 3in]{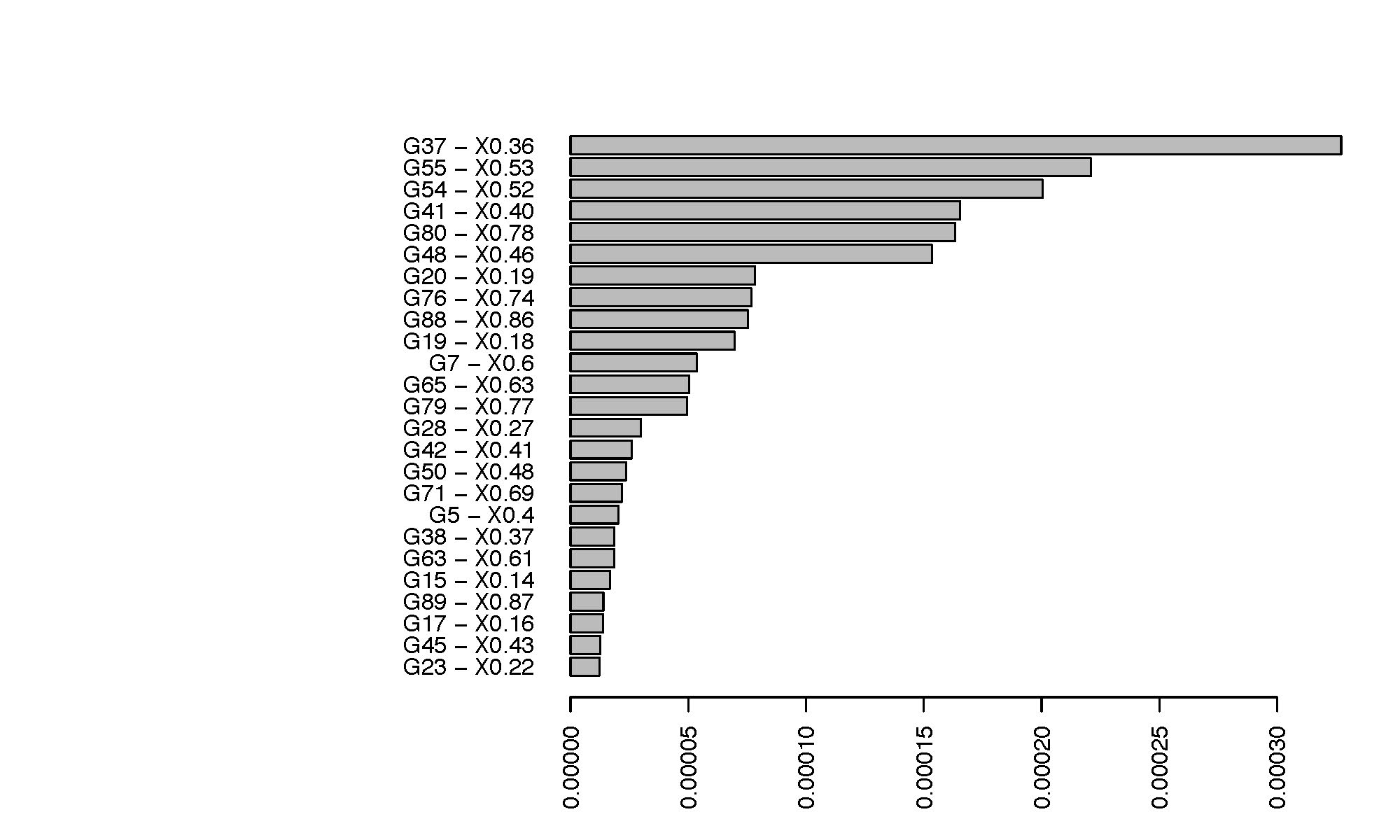}
}
\caption{Bar graph showing the approximate individual posterior probabilities of
genera.  The labels on the left represent a coding of the genera.}
\label{finalprobs}
\end{figure}

With that in mind, we note that the top four clique loglinear models for genera from Section \ref{sec:analyses} represent a very high -- around 98-99\% -- of the posterior probability.  Thus, as a verification that the Occam's window is not loosing too much of the posterior probability, we note that Figure \ref{finalprobs} shows the 5-th through 29-th individual posterior probabilities of the various genera.  The labeling on the y-axis is just a coding we used for the genera and is not of concern.  What is of concern is that with the 1133 models, the top 4 individual probabilities have something like 99\% of the posterior probability and the probability of something unknown is around .08. This is also discussed in Section \ref{sec:discussion}. Figure \ref{finalprobs} simply shows that including more models in the Occam's window BMA would not make much difference.

Thus, although we have not looked at acceptance probabilities, our methodology has generated a
collection of models that captures the vast majority of the posterior probability so the Occam's windowing 
is not losing too much information.  Otherwise put, it seems as though the 1133 models resulting
from our search procedure provides a reasonable approximation to the actual posterior. In our view acceptance probabilities are only interesting if one believes the model list is physically meaningful and issues of sparsity are not relevant. Indeed, sparsity will tend to force bigger jumps (when they occur) leading to an erratic pattern of acceptance rates where non-sparsity will tend to be continuous. The relationship between sparsity and acceptance rates -- while interesting -- has not been well investigated as far as we know and is beyond our present scope.

\subsection{Skewed data and the BIC}
\label{sec:skewed}

The BIC is not an approximation to the posterior:  it is an approximation to the mode of the posterior and the location of the mode indicates a good model (and arguably a good parameter value). The detailed behavior of the BIC (or any information based model selection criterion) is a general problem that will not likely be resolved in our lifetimes and we are not sure how to verify conclusively that the BIC convergence our method needs to hold since the true models are unknown (and quite possibly unknowable given the dynamic nature of organisms).

Here, then, is the state of play: the BIC converges for most well-behaved distributions with sample sizes not too different from those required for analogous convergences for the 
CLT.  After all, posterior convergence (when conditioning on all the data) is much at one with the frequentist CLT.  So, as a generality, for some linear regression type models that are 
similar to loglinear models you want around 30 data points/parameter to be assured of good convergence. This assumes the models on the model list are not too similar. We have a 
high dimensional contingency table with $2^{\#({\sf genomes)}}$ where $\#({\sf genomes})$ is in the thousands. Thus, in our first real world example (see Section \ref{sec:analyses}) we used 95 genera.  As noted there were 377 cells with strictly positive counts.  Loosely, therefore, we have 377 sets of parameters (the cardinality of the sets of parameters being due to the hierarchical structure of the loglinear models). In terms of cell counts, the largest was 332k and the second largest was 11.6k and there are a lot that are very small \--- see Figure \ref{poscellcts}.  Convergence will be determined by the total cell count which in turn is largely determined by the depth of sequencing.

\begin{figure}
\centering
\centerline{
\includegraphics[width = 5in, height= 2.5in]{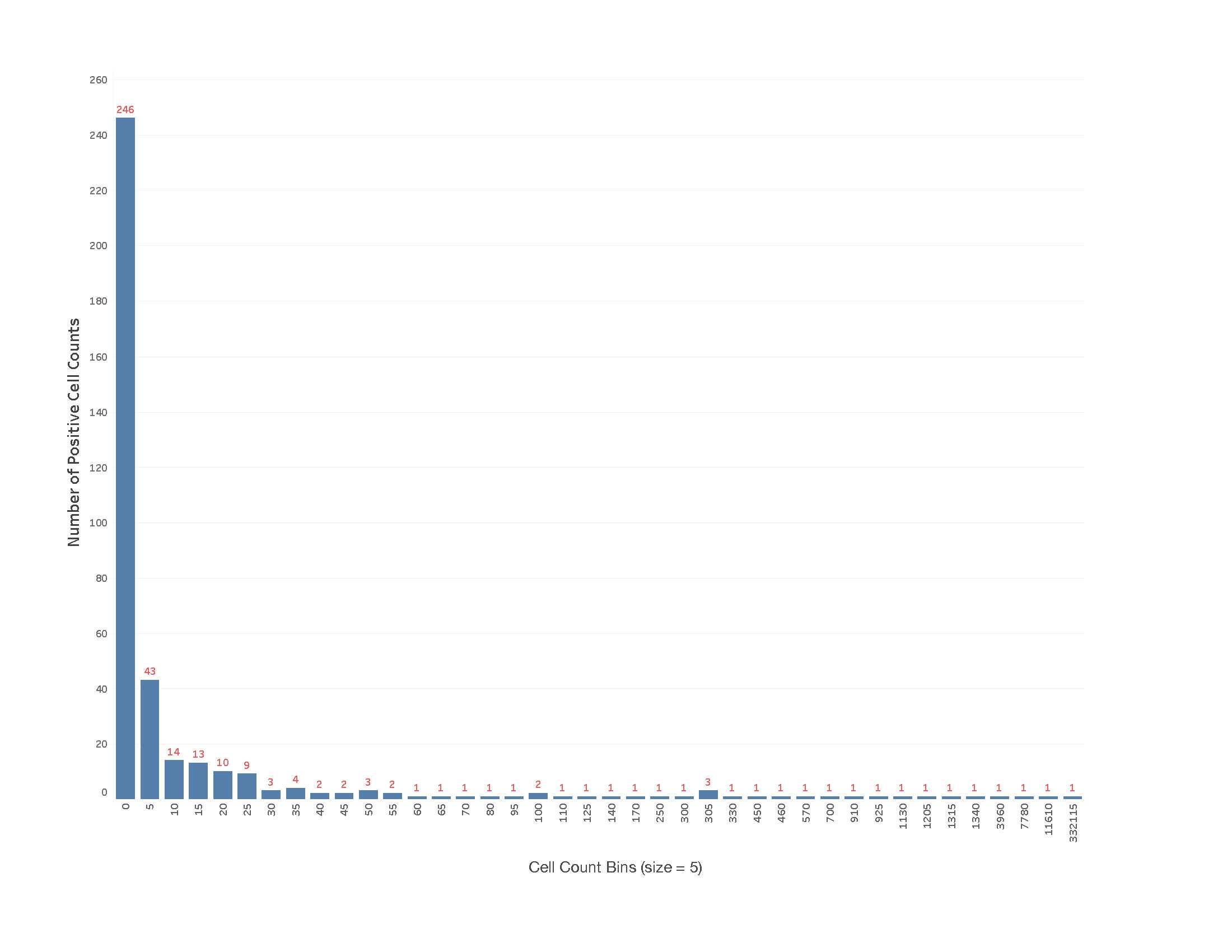}
}
\caption{Histogram of the positive cell counts for the HMP data.}
\label{poscellcts}
\end{figure}

It is obvious that the positive cell counts are unbalanced and strongly skewed; this is typical for high dimensional contingency tables and indicates sparsity. It is also well known that loglinear models might not be able to capture the unbalancedness well because the expected cell counts under a loglinear model must be strictly positive for all the cells including the ones with an observed count of zero. Since the grand total of the table of expected counts must equal the grand total of the observed table, the total sum of counts gets redistributed to the many cells with an observed count of zero. Hence what you see is that the expected counts for the cells with the largest observed counts are much smaller than those observed counts. 

As such, loglinear models might fail to properly capture the imbalance of the cell counts. However, capturing the magnitude of the cell counts is not our end goal. Instead, our target is the determination of the interaction structure that exists among the variables in the data. Thus, the sharp skewness in Figure \ref{poscellcts} has nothing to do with the BIC per se,
i.e., with the way we select loglinear models. Informally, the $n$ in the BIC is the grand total of the observed table. Since the sum of the counts is in the hundreds of thousands, we hope that even given the unbalancedness of the cell counts, the BIC will perform well. \cite{schwarz-1978} established an optimality property of the BIC in terms of hypothesis testing, i.e., decision making, so it is not at all clear that there is a better choice for a model selection principle. Moreover, to give an indication of the scale of how well BIC works, \cite{lv-liu-2014} examine the use of the BIC in a model mis-specification setting involving logistic regression, and show the BIC works reasonably well even when $n=200$ and $p=1000$.

The strong, clear connectivities are likely the ones where the BIC is giving  good approximation.  This is indicated by the higher posterior probabilities in Figure \ref{fig:generaconvergence}. The other connections may be weaker due to the connectivity being genuinely weaker or due to the BIC giving a poor approximation to the posterior mode.  Thus, our method identifies connections that we are pretty sure are present given the data.  In addition, our method suggests other connections may be worth exploring. Indeed, our example from Section \ref{sec:footwound} which is further developed in Section \ref{sec:wound-supp} shows that our methodology can potentially save a lot of expensive lab work.

\subsection{Robustness and sparsity}
\label{sec:robsparse}

Statistical analyses of the sort we are proposing have the right sort of robustness to be credible. In our view, the use of the BMA increases robustness and the restriction to clique
loglinear provides the required robustness.  The lack of robustness is due to sparsity and this is typical with sparse methods.

To see what happens we did simulations. In the connectivity matrix for the simulation example from Section \ref{sec:sims}, we changed 0s to 1s using a Bernoulli with $p= .01, .05, .1$, and looked at the downstream effects on the models and the resulting graphs from the BMAs. We did this for 1000 reads total, as in the original simulation example, and for 5000 reads total. Figure \ref{MCMCoutput} is the analog of Figure \ref{fig:generaconvergence} for the performance of our stochastic search algorithm.  As can be seen, as the number of 1s increases, the steps smooth out. This holds for 1000 total reads and 5000 total reads. This suggests that as sparsity decreases a single model becomes more and more reasonable.  We suspect that this is an artifact of simply having more reads and hence more interaction terms. In the limit, all possible interactions terms will be present and this is not helpful.

\begin{figure}
\begin{center}
\begin{tabular}{cc}
\includegraphics[height = 3.5cm]{convergenceGenera.jpg} &
\includegraphics[height = 3cm]{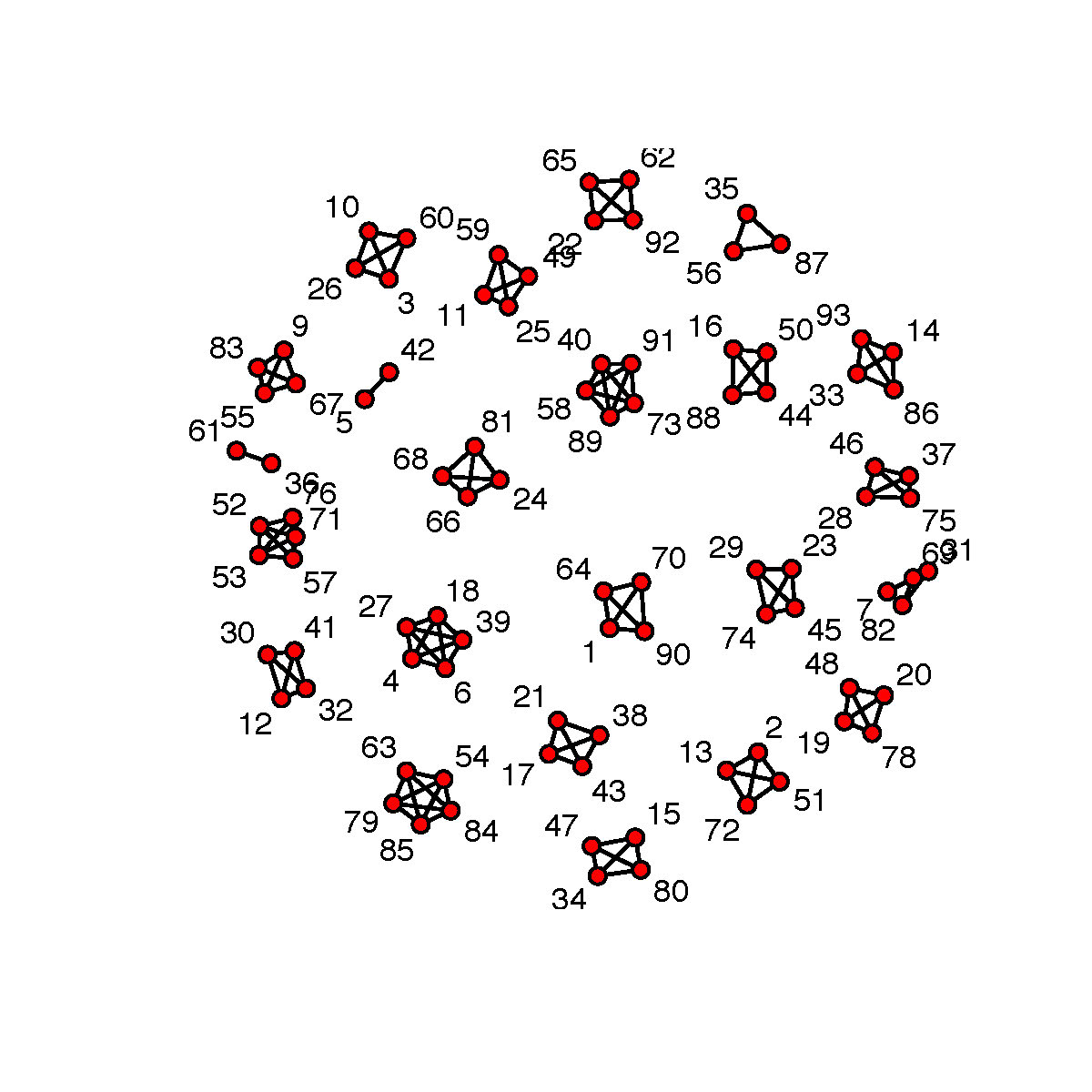} \\
\includegraphics[height = 3cm]{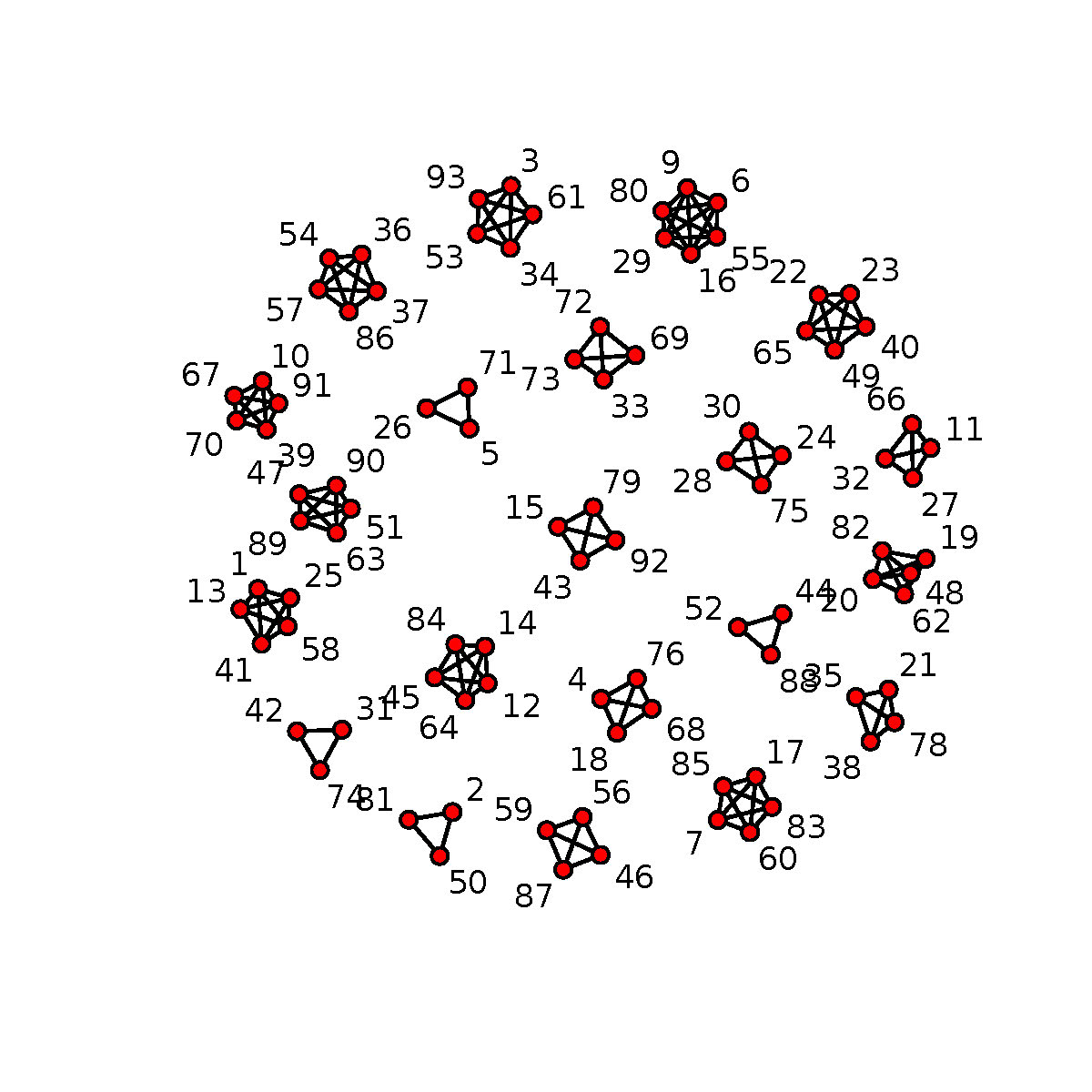} &
\includegraphics[height = 3cm]{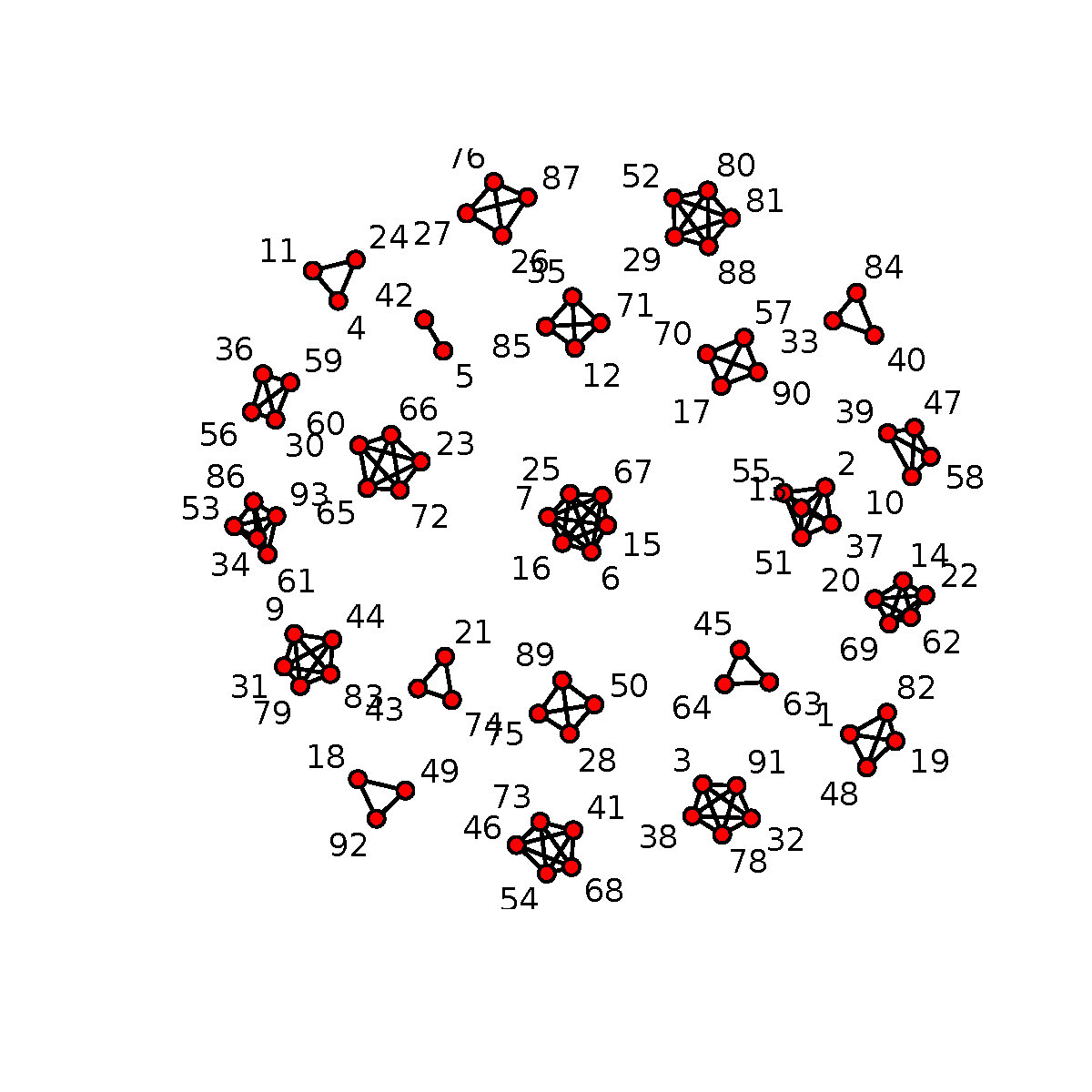}
\end{tabular}
\end{center}
\caption{The upper left graph is analogous to Figure \ref{fig:generaconvergence}.  The others correspond to adding 1's to the connectivity matrix using Bernoulli distributions with $p=.01$, top right; $p=.05$, lower left, and $p=.1$ lower right.}
\label{MCMCoutput}
\end{figure}

\section{Simulations}
\label{sec:sims-supp}

To verify the performance of the proposed method, a synthetic experiment was created with a known bacterial independence structure. This constitutes our ground truth. A number of 2,273 bacterial genomes from the National Center for Biotechnology Information (NCBI) GenBank database were obtained from the complete set of genomes.  These complete genomes are considered to be very high quality by GenBank, and are deemed to have a final DNA sequence for their respective genomic sequences (chromosomes and/or plasmids).

\begin{figure}[H]
\centering
\includegraphics[width=0.65\textwidth]{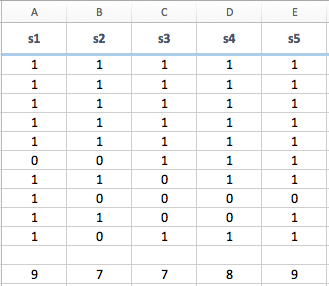}
\caption{\label{fig:known_matrix}Input matrix with a predefined independence structure among the genomes (columns) based on their sequencing reads (rows) and column totals (last row).}
\end{figure}

A binary matrix that dictates the reads (rows) that are shared by each genome (columns) was generated.  Figure \ref{fig:known_matrix} illustrates such a matrix.  In this matrix, rows represent sequencing reads, columns represent known bacterial genomes, and cells contain a 1 when a read maps to a given genome, and 0 otherwise.  The last row of the matrix contains the number of reads that match a given genome, and we use this total match count as the starting point to create a independence among a set of genomes.

A simulation script was developed in Python (version 2.7.10) to facilitate the creation of random binary matrices.  The script takes as input the total number of genomes to include in the experiment, the maximum clique size, the total number of reads, and a percentage of the genomes to leave out of the connections.

The script will then randomly select the genomes to make connections based on shared reads and construct cliques of size two to the maximum number specified in the parameters. Genomes are essentially binary vectors where each location is a read, and a 1 or 0 at that location marks a read as mapping to the genome.  Creating linked genomes is a matter of establishing binary vectors that share the same amount of reads at similar positions.  Genomes with identical sequences will have identical read catalogs represented by identical binary vectors. As output the simulation script creates the binary matrix that represents the known independence graph and is used as input to the software \--- see Section \ref{sec:software}.

\section{Wound microbiome}
\label{sec:wound-supp}
The wound microbiome experiment investigated the bacterial population involved in chronic wound healing in elderly diabetic patients.  The microbial community in the wound bed, wound edge, and peripheral healthy skin was compared at two time points.

Samples were collected from the wound bed and wound edge, as well as swab samples from the peripheral healthy skin of 10 patients.  These samples were collected twice: once for the patient's first visit before treating and redressing the wound, and in the second visit one week later. A total of 50 samples were collected. They were sequenced using the standard Illumina protocol for 16S rRNA gene sequencing, using the Illumina MiSeq instrument.  After processing by the company (\textit{Second Genome}), the investigators shared with us a table of already calculated OTU abundances.  We did not make the jump from 16S to OTUs; it was already done and we just used the data.   We have a copy of the report from the company but are not yet authorized to release it.  What we can say is that an OTU is a cluster of reads with a similar 16S-gene sequence. Usually the identity threshold is 96-97\% or higher. The basic idea is that similar bacteria will have similar 16S-gene sequences, and they can be identified by grouping them together given a threshold: 96-97\% agreement usually gives resolution at the genera level, 98-99\% usually gives approximate resolution at the species level. It is unclear if OTUs are useful at this time for finer levels such as strains. Thus, for present purposes, it is enough to observe that the 16S rRNA sequencing data was converted to operational taxonomic units (OTU) and  analyzed to identify the biologically significant OTUs. 

\subsection{OTU analysis}

The analysis goal was to identify the OTUs that are significant between patients whose wound healed (``healer'') versus those patients that did not heal (``non-healer'') in the context of the location of the wound: the wound bed, wound edge, or peripheral healthy skin.

\begin{verbatim}
#	Load data
otu <- read.csv('otuTable_counts.csv', header=T, row.names=1)
samp <- read.csv('otuTable_samples.csv', header=T, row.names=1, 
                              stringsAsFactors=FALSE)
otu.mat <- as.matrix(otu)

source("https://bioconductor.org/biocLite.R")
biocLite()
biocLite('phyloseq')
library("phyloseq")
OTU <- otu_table(otu.mat, taxa_are_rows = TRUE)


#	Phyloseq
physeq<-phyloseq(OTU)
sampledata<-sample_data(samp)
physeq1<-merge_phyloseq(physeq,sampledata)
biocLite('edgeR')


#	EdgeR Analysis
library(edgeR)
Diagnosis<-get_variable(physeq1, "type")
Location<-get_variable(physeq1, "location")
design<-model.matrix(~Diagnosis + Location)
x<-as(otu_table(physeq1), "matrix")+1L
x<-DGEList(counts=x, group=Diagnosis)


#	Calculate norm factors and estimate dispersion
x<-calcNormFactors(x, method="RLE")
x<-estimateGLMCommonDisp(x, design)


#	Model fitting
fit <- glmFit(x, design)
lrt <- glmLRT(fit)
\end{verbatim}

\subsection{OTU results}
185 OTUs were identified as being significant between ``healer'' and ``non-healer'' patients.  Table \ref{significant_otus} contains the results. Figure \ref{fig:venn_diagram} contains the intersection between the resulting lists from Table \ref{significant_otus}.

\begin{table}[H]
\centering
\caption{Significant OTUs between healing and non-healing patients.}
\label{significant_otus}
\begin{tabular}{l|c}
\multicolumn{1}{c|}{\textbf{Wound Location}} & \textbf{Number of Significant OTUs} \\ \hline
Bed                                          & 144                                 \\
Edge                                         & 100                                 \\
Peripheral Healthy                           & 75                                 
\end{tabular}
\end{table}

\begin{figure}[H]
\centering
\includegraphics[width=0.5\textwidth]{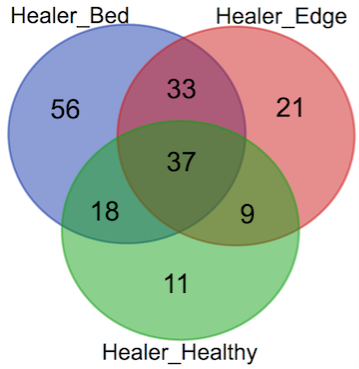}
\caption{\label{fig:venn_diagram}Intersection between the significant OTUs from the wound bed, wound edge, and peripheral healthy skin.}
\end{figure}

To examine the dependencies among the three wound locations, we ran the significant OTUs to create a independence graph using the proposed method.  At the same time, we ran a principal component analysis (PCA) on the samples, and also performed hierarchical clustering to see if we could identify groups of samples that were linked together by their respective OTUs.  Figure \ref{fig:wound_panels} contains the dendrogram plot for the hierarchical clustering (panel A), the PCA plot (panel B), and the independence graph (panel C).

\begin{figure}[H]
\centering
\includegraphics[width=0.95\textwidth]{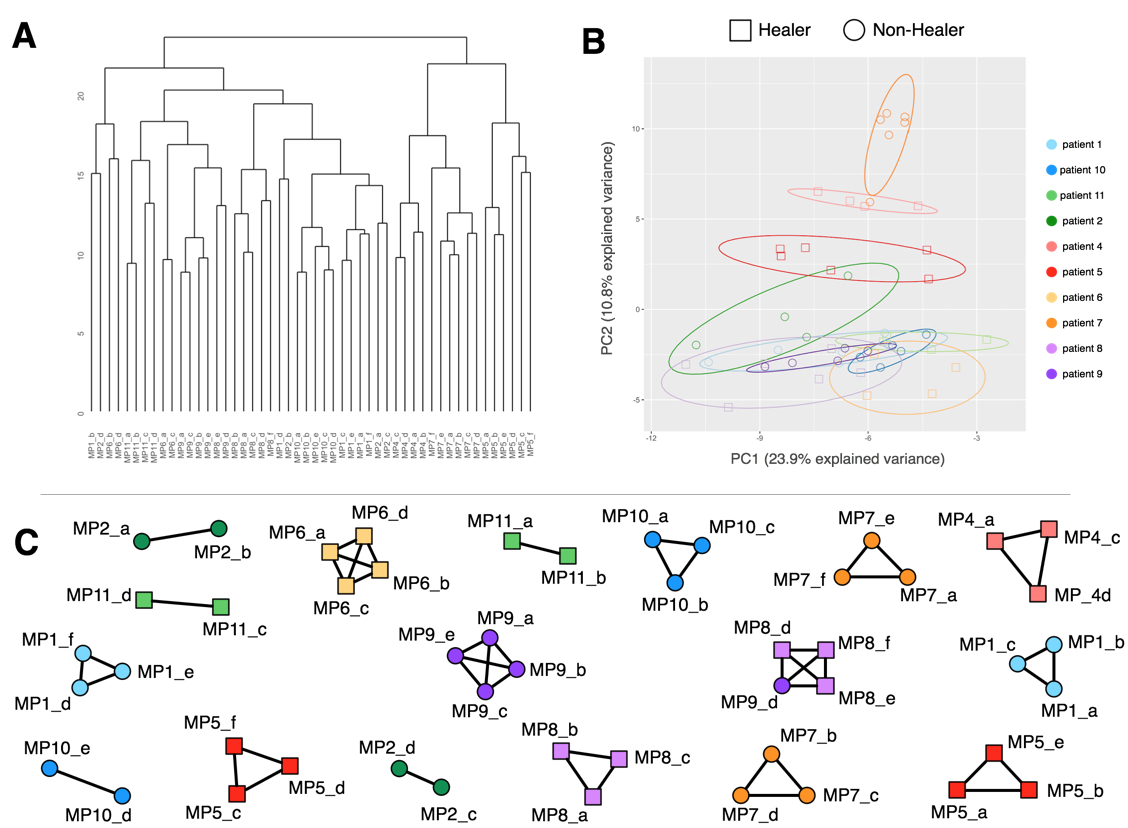}
\caption{\label{fig:wound_panels}Significant OTUs in the wound Bed location. Sample clustering (panels A and B) and independence graph (panel C).}
\end{figure}

In all figures, the patient effect is the strongest differentiating factor as samples from a given patient tend to cluster together in hierarchical clustering (panel "A") and PCA (panel "B").  The independence graph creating by the model is also in agreement with the plots as samples from a given patient, or a patient condition (healer or non-healer), are connected in a clique.

\section{Software availability}
\label{sec:software}

The software for this project is open source software, available under the GNU General Public License, Version 3.  The software is developed in R, version 3.2.3, and can be obtained at the following GitHub repository:\\\

\url{https://github.com/camilo-v/Clique_Log_Linear}


\end{document}